\documentclass[11pt]{article}

\usepackage{tabularx,booktabs,multirow,delarray,array}
\usepackage{graphicx,amssymb,amsmath,amssymb}
\usepackage{latexsym}
\usepackage{wrapfig}
\usepackage{enumerate}
\usepackage{todonotes}

\def\qed{\hbox{\rlap{$\sqcap$}$\sqcup$}}
\def\calP{\mathcal{P}}
\def\calT{\mathcal{T}}
\def\calE{\mathcal{E}}
\def\calR{\mathcal{R}}
\def\calF{\mathcal{F}}

\newtheorem{propo}{Proposition}

\newtheorem{cor}{Corollary}
\newtheorem{theorem}{Theorem}[section]
\newtheorem{lemma}{Lemma}[section]
\newenvironment{proof}{\par\noindent{\bf Proof:}}{\mbox{}\hfill$\qed$\\}

\newcommand{\ignore}[1]{ }

\begin{document}
\title{A Polynomial Time Algorithm to Compute an Approximate Weighted Shortest Path}

\author{R. Inkulu\thanks{R.~Inkulu's research was supported in part by NBHM under grant 248(17)2014-R\&D-II/1049.}\footnote{
  Department of Computer Science \& Engineering; IIT Guwahati, India; \texttt{rinkulu@iitg.ac.in}
}
\and
Sanjiv Kapoor\footnote{
  Department of Computer Science; IIT Chicago, USA; \texttt{kapoor@iit.edu}
}}
\date{}
\maketitle

\begin{abstract}
We devise a polynomial-time approximation scheme for the classical geometric problem of finding an approximate short path amid weighted regions.
In this problem, a triangulated region $\calP$ comprising of $n$ vertices, a positive weight associated with each triangle, and two points $s$ and $t$ that belong to $\calP$ are given as the input.
The objective is to find a path whose cost is at most $(1+\epsilon)$OPT where OPT is the cost of an optimal path between $s$ and $t$.
Our algorithm initiates a discretized-Dijkstra wavefront from source $s$ and progresses the wavefront till it strikes $t$.
This result is about a cubic factor (in $n$) improvement over the Mitchell and Papadimitriou '91 result \cite{journals/jacm/MitchellP91}, which is the only known polynomial time algorithm for this problem to date.
Further, with polynomial time preprocessing of $\calP$, a map is computed which allows answering approximate weighted shortest path queries in polynomial time.
\end{abstract}

\section{Introduction}
\label{sect:intro}

The problem of computing a shortest path in polygonal subdivisions is important and well-studied due to its applications in geographic information systems, VLSI design, robot motion planning, etc.
A survey of various shortest path problems and algorithms may be found in Mitchell \cite{hb/cg/Mitch98}.
In this paper, we devise an algorithm for the {\it weighted shortest path problem} \cite{journals/jacm/MitchellP91}: given a triangulation $\calP$ with $O(n)$ faces, each face associated with a positive weight, find a path between two input points $s$ and $t$ (both belonging to $\calP$) so that the path has minimum cost among all possible paths joining $s$ and $t$ that lie on $\calP$. 
The cost of any path $P$ is the sum of costs of all line segments in $P$, whereas the cost of a line segment is its Euclidean length multiplied by the weight of the face on which it lies.
The weighted shortest path problem helps in modeling region specific constraints in planning motion.

To compare the time complexities of various algorithms in the literature, we use the following notation -
$n$: number of vertices defining $\calP$; $l_{max}$: length of the longest edge bounding any face of $\calP$; $N$: maximum coordinate value used in describing $\calP$; $w_{max}$: maximum non-infinite weight associated with any triangle; $w_{min}$: minimum weight associated with any triangle; $\theta_{min}$: minimum among the internal face angles of $\calP$; and, $\mu$: ratio of $w_{max}$ to $w_{min}$.
(Note that the same notation is used in later parts of the paper as well.)

Mitchell and Papadimitriou \cite{journals/jacm/MitchellP91} presented an algorithm that finds an approximate weighted shortest path in $O(n^8 \lg{\frac{nN\mu}{\epsilon}})$ time.
Their algorithm essentially builds a shortest path map for $s$ by progressing continuous-Dijkstra wavefront in $\calP$ using the Snell's laws of refraction.
By introducing $m$ equi-spaced Steiner points on each edge of $\calP$ and building a graph spanner over these points, Mata and Mitchell \cite{conf/compgeom/MataM97} devised a preprocessing algorithm to construct a graph spanner in $O(kn^3)$ time, where $k = O(\frac{\mu}{\epsilon \theta_{min}})$, on which $(1+\epsilon)$-approximate weighted shortest path queries are performed.
Lanthier et~al. \cite{journals/algorithmica/LanthierMS01} independently devised an $O(n^5)$ time approximation algorithm with an additive error of $O(l_{max} w_{max})$ by choosing $m = n^2$.
Instead of uniform discretization (as in \cite{conf/compgeom/MataM97}), Aleksandrov et~al. \cite{conf/swat/AleksandrovLMS98,journals/jacm/AleksandrovMS05} used logarithmic discretization and devised an $O(\frac{nN^2}{\sqrt{\epsilon}}\lg(\frac{NW}{w})\lg{\frac{n}{\epsilon}}\lg{\frac{1}{\epsilon}})$ time approximation algorithm.
Sun and Reif \cite{journals/jal/SunR06} provided an approximation algorithm, popularly known as BUSHWHACK, with time complexity $O(\frac{nN^2}{\epsilon}\lg({N\mu})\lg{\frac{n}{\epsilon}}\lg{\frac{1}{\epsilon}})$.
Their algorithm dynamically maintains for each Steiner point $v$, a small set of incident edges of $v$ that may contribute to an approximate weighted shortest path from $s$ to $t$.
More recently, Cheng et al. \cite{conf/soda/ChengJV15} devised an approximation algorithm that takes $O(\frac{kn + k^4 \lg(k/\epsilon)}{\epsilon}\lg^2{\frac{\rho n}{\epsilon}})$ time, where $k$ is the smallest integer such that the sum of the $k$ smallest angles in $\calP$ is at least $\pi$.
The query version of this problem is addressed in \cite{conf/swat/AleksandrovLMS98,journals/jacm/AleksandrovMS05,journals/jal/SunR06,journals/jacm/MitchellP91,journals/siamcomp/ChengNVW10,journals/dcg/AleksandrovDGMNS10}.
Further, algorithms in Cheng et~al. \cite{journals/siamcomp/ChengNVW10} handle the case of measuring the cost of path length in each face with an asymmetric convex distance function. 
%\todo{Any more recent works}

\subsubsection*{Our contribution}

The time complexities of each of the above mentioned solutions, except for \cite{journals/jacm/MitchellP91}, are polynomial in $n$ as well as in parameters such as $\epsilon, \theta_{min}, N$ and $\mu$.
Hence, strictly speaking, these algorithms are not polynomial.
Like \cite{journals/jacm/MitchellP91}, this paper devises an algorithm that is polynomial in time complexity.
The time complexity of our algorithm is $O(n^5(\lg{\frac{n}{\epsilon}})(\lg{\frac{\mu}{\sqrt{\epsilon}}}))$, which is about a cubic factor improvement from \cite{journals/jacm/MitchellP91}.
As established in \cite{journals/jacm/MitchellP91}, there are  $\Omega(n^4)$ events that need to be handled in order to find the interaction of shortest path map with the $\calP$.
Our algorithm takes first steps to provide a solution that is sub-quadratic in the number of events.

This result uses several of the characterizations from \cite{journals/jacm/MitchellP91} in order to design a simple and more efficient algorithm.
Every ray is a simple path in $\mathcal{P}$.
Our approach discretizes the wavefront in the continuous-Dijkstra's approach by a set $S$ of rays whose origin is source $s$.
These rays are distributed uniformly around $s$.
As the discrete wavefront propagates, a subset of the rays in $S$ are progressed (traced) further while following the Snell's laws of refraction.
Each of these subsets of rays is guided by two extreme rays from that subset, i.e., all the rays in that subset lie between these two special rays. 
Essentially, each such subset represents a section of the wavefront and is called a {\em bundle}.
For any vertex $v$ in $\calP$, whenever such a section of the wavefront strikes a vertex $w$, we initiate another discrete-wavefront (set of rays) from $w$.
We continue doing this until the wavefront (approximately) strikes $t$.
In summary our contributions are:
\begin{enumerate}
\item A discretized approach to propagating wavefronts using bundles.
\item
An algorithm that computes an $(1+\epsilon)$-approximate weighted shortest path from $s$ to $t$ in $O(n^5(\lg{\frac{n}{\epsilon}})(\lg{\frac{\mu}{\sqrt{\epsilon}}}))$ time.
\item
Further, we preprocess $\calP$ in $O(n^5(\lg{\frac{n}{\epsilon}})(\lg{\frac{\mu}{\sqrt{\epsilon}}})(\lg{N}))$ time to compute a data structure for answering single-source approximate weighted shortest path queries in $O(n^4(\lg{\frac{n}{\epsilon}})(\lg{\frac{\mu}{\sqrt{\epsilon}}})(\lg{N}))$ time.
\end{enumerate}

Section~\ref{sect:prelim} lists relevant propositions from \cite{journals/jacm/MitchellP91} and \cite{journals/jal/SunR06}, and defines terminology required to describe the algorithm. 
Section~\ref{sect:algooutline} outlines the algorithm while introducing few structures used in the algorithm.
In Section~\ref{sect:boundrays}, we bound the number of rays.
The details of the algorithm are provided in Section~\ref{sect:algodetails}.
Section~\ref{sect:interpol} describes an interpolation scheme to improve the time complexity of the algorithm. 
Section~\ref{sect:analysis} argues for the correctness and analyzes the time complexity of the algorithm.
And, the conclusions are given in Section~\ref{sect:conclu}.

\section{Preliminaries}
\label{sect:prelim}

We define the problem using the terminology from \cite{journals/jacm/MitchellP91}.

We assume a planar subdivision ${\cal P}$, that is polygonal and specified by triangular faces.
Each face $f$ has a weight $w_f$ associated with it. 
We denote the weight of a face $f$ (resp. edge $e$) with $w_f$ (resp. $w_e$).
For an edge $e$ shared by faces $f'$ and $f''$, the weight of $e$ is defined as $\min(w_{f'}, w_{f''})$.

A {\it path} is a continuous image of an interval, say $[0, 1]$, in the plane.
A {\it geodesic path} is a path that is locally optimal and cannot, therefore, be shortened by slight perturbations.
An {\it optimal path} is a geodesic path that is globally optimal.
The general form of a weighted geodesic path is a simple (that is, not self-intersecting) piecewise linear path that goes through zero or more vertices while possibly crossing a zero or more edges.
The Euclidean length of a line segment $l$ is denoted by $\Vert l \Vert$.
Let $P$ be a geodesic path with line segments $l_1, l_2, \ldots, l_k$ such that $l_i$ lies on face $f_i$, for every $i$ in $[1, k]$; then the {\it weighted Euclidean distance} (also termed the {\it cost}) of path $P$ is defined as the $\sum_i \Vert l_i \Vert w_{f_i}$.\\ \\
{\bf Weighted Shortest Path problem:}
{\em Given a finite triangulation $\calP$ in the plane with two points $s$ (source) and $t$ (destination) located on $\calP$, an assignment of positive integral weights to faces of $\calP$, and an error tolerance $\epsilon$, the {\it approximate weighted shortest path} problem is to determine a path $P$ from $s$ to $t$ that lies on $\calP$ such that the cost of $P$ is at most $(1+\epsilon)$ times the cost of an optimal path from $s$ to $t$.}\\

We assume that all input parameters of the problem are specified by integers.
In particular, all vertices have non-negative integer coordinates.
Further, we assume that $s$ and $t$ are two vertices of $\calP$.
A {\it ray} in our algorithm is a piecewise linear simple path in $\mathcal{P}$ such that endpoints of each of the line segments that it contains lie on the edges of $\calP$.
Let $f'$ be the first face traversed by a ray $r'$ and let $f''$ be the first face traversed by a ray $r''$.
The angle between $r'$ and $r''$ is defined as the angle between the vectors induced by $f' \cap r'$ and $f'' \cap r''$.
Unless specified otherwise, all angles are acute.
Let $\overrightarrow{v}$, $\overrightarrow{v_1}$ and $\overrightarrow{v_2}$ be three unit vectors that originate from a point $a$ in $\mathbb{R}^2$ such that the vectors $\overrightarrow{v_1}, \overrightarrow{v_2}$ lie in the same half-plane defined by the line induced by $\overrightarrow{v}$.
The {\it cone} $C(a, \overrightarrow{v_1}, \overrightarrow{v_2})$ is the set comprising of all the points that are positive linear combinations of $\overrightarrow{v_1}$ and $\overrightarrow{v_2}$.

\subsection*{A few facts from the literature} 

First we list few propositions, definitions and descriptions from Mitchell and Papadimitriou \cite{journals/jacm/MitchellP91} tailored for our purpose.

A sequence of edge-adjacent faces is a list, $(f_1, f_2, \ldots, f_{k+1})$, of two or more faces such that, for every $i$, face $f_i$ shares edge $e_i$ with face $f_{i+1}$.
Further, the corresponding sequence of edges $\calE = (e_1, e_2, \ldots, e_k)$ is referred as an {\it edge sequence}.
When a geodesic path $P$ crosses edges in $\calE$ in the order specified by $\calE$ and without passing through any vertex, then $\calE$ is the {\it edge sequence of path $P$}.
A geodesic path $p_0, p_1, \ldots, p_k$ in $\mathcal{P}$, is termed a {\it ray} $r$ as it behaves similar to a  ray of light.
Further, when the geodesic path is extended to add a  point $p_{k+1}$ in $\mathcal{P}$, by the  line segment $p_kp_{k+1}$, the ray $r$ is said to be {\it traced} (or, {\it progressed}) to $p_{k+1}$.
Any point or line segment in $\cal{P}$ from which at least one ray is initiated is termed a {\it source}.

	\begin{wrapfigure}{r}{0.5\textwidth}
	\centering
	\begin{minipage}[t]{\linewidth}
	\centering
	\includegraphics[totalheight=0.65in]{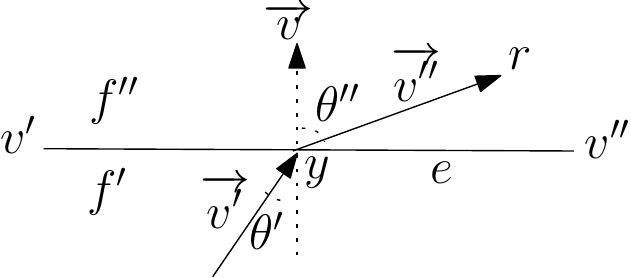}
	\vspace{-0.15in}
	\caption{\footnotesize Illustrating refraction}
	\label{fig:refraction}
	\end{minipage}
	\end{wrapfigure}

Let $f'$ and $f''$ be two faces with shared edge $e(v', v'')$.
(See Fig. \ref{fig:refraction}.)
Let $l', l''$ be two successive line segments along a ray $r$ with $l'$ lying on face $f'$ and $l''$ lying on face $f''$ with point $y \in l' \cap l'' \cap e$.
Let $\overrightarrow{v'}$ be a vector entering $y$ and containing $l'$ and let $\overrightarrow{v''}$ be a vector with origin $y$ and containing $l''$.
Also, let $\overrightarrow{v}$ be a vector normal to edge $e$, passing through point $y$ to some point in face $f''$. 
The angle $\theta'$ between $\overrightarrow{v}$ and $\overrightarrow{v'}$ is known as the {\it angle of incidence} of $r$ onto $e$.
And, the angle $\theta''$ between $\overrightarrow{v}$ and $\overrightarrow{v''}$ is known as the {\it angle of refraction} of $r$ from $e$.

	\begin{wrapfigure}{r}{0.5\textwidth}
	\centering
	\begin{minipage}[t]{\linewidth}
	\centering
	\includegraphics[totalheight=0.7in]{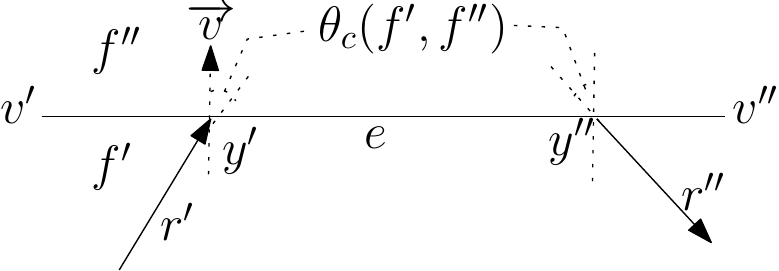}
	\caption{\footnotesize Illustrating critical reflection}
	\label{fig:critrefl}
	\end{minipage}
	\end{wrapfigure}

When $w_{f'} > w_{f''}$, the {\it critical angle of $e$}, denoted by $\theta_c(f', f'')$ is $\sin^{-1}(\frac{w_{f''}}{w_{f'}})$.
(When $f'$ and $f''$ are understood from the context, the critical angle of $e$ is also denoted by $\theta_c$.)
(See Fig. \ref{fig:critrefl}.)
If $\theta' < \theta_c(f', f'')$, then the angles $\theta'$ and $\theta''$ are related by Snell's law of refraction with $w_{f'} \sin{\theta'} = w_{f''} \sin{\theta''}$.
Since there does not exist a geodesic weighted shortest path with $\theta' > \theta_c(f', f'')$, we only need to consider the case in which $\theta'$ equals to $\theta_c(f', f'')$. 
Let $y'$ be the point of incidence of ray $r'$ from face $f'$ onto $e$ with angle of incidence $\theta_c(f', f'')$.
Also, let $\overrightarrow{v}$ be a vector normal to edge $e$, passing through point $y'$ to some point in face $f''$.
Then the geodesic path travels along $\overrightarrow{y'v''}$ for some positive distance before exiting edge $e$ back into face $f'$, say at a point $y''$ located in the interior of $y'v''$, while making an angle $-\theta_c(f', f'')$ with $\overrightarrow{v}$. 
We say that the path $P$ is {\it critically reflected} by edge $e$ and the line segment $y'y''$ is termed a {\it critical segment} of path $P$ on $e$.
Sometimes, we also say $y'v''$ is a critical segment corresponding to the critcial incidence of a ray at $y'$ from face $f'$.
The point $y'$ (the closer of the two points $\{y', y''\}$ to $s$) is known as a {\it critical point of entry} of path $P$ to edge $e$ and $y''$ is known as the corresponding {\it critical point of exit} of path $P$ from edge $e$.

The following propositions from the literature are useful for our algorithm.

\begin{propo}[Lemma~3.7, \cite{journals/jacm/MitchellP91}]
\label{prop:betwcrit}
Let $P$ be a geodesic path.
Then either (i) between any two consecutive vertices on $P$, there is at most one critical point of entry to an edge $e$, and at most one critical point of exit from an edge $e'$ (possibly equal to $e$); or
(ii) the path $P$ can be modified in such a way that case (i) holds without altering the length of the path. 
\end{propo}

Let $f'$ and $f''$ be two faces with shared edge $e$.
For any point $x \in e$, a {\it locally $f'$-free path} strikes $x$ from the exterior of face $f'$ and is locally optimal.

\begin{propo}[Lemma~7.1, \cite{journals/jacm/MitchellP91}]
\label{prop:edgeseqlen}
For a face $f$ of $\cal{P}$, let $P$ be a shortest locally $f$-free path.
Let $P'$ be a sub-path of $P$ such that $P'$ goes through no vertices or critical points.
Then, $P'$ can cross an edge $e$ at most $O(n)$ times.
Thus, in particular, the cardinality of any edge sequence of path $P$ is $O(n^2)$.
\end{propo}

\begin{propo}[Lemma~7.4, \cite{journals/jacm/MitchellP91}]
\label{prop:numcritsrc}
There are at most $O(n)$ critical points of entry on any given edge $e$.
\end{propo}
Finally we have a  proposition that provides us the {\it non-crossing property of (weighted) shortest paths}.
\begin{propo}[Lemma~1, \cite{journals/jal/SunR06}]
\label{prop:noncrossing}
Any two geodesic shortest paths that originate from the same point in $\cal{P}$ cannot intersect in the interior of any weighted region of $\cal{P}$.
\end{propo}

\section{Algorithm outline}
\label{sect:algooutline}

In our algorithm, we progress a discretized-Dijkstra wavefront, as the traditional approach of
progressing a continuous-Dijkstra wavefront results in a complicated algorithm, 
This wavefront is defined and expanded using rays.
Every {\it ray} is a geodesic path in $\calP$ from its point of origin to a point on the wavefront.
The wavefront that defines the locus of points at weighted Euclidean distance $d$ from the source is identified by points on the the rays that are at weighted Euclidean distance $d$ from the source.
Note that geodesic paths in ${\cal P}$ are piece-wise linear.
We {\it initiate} a set $\calR(s)$ of rays whose origin is $s$.
The rays in $\calR(s)$ are ordered according to their counterclockwise angle with the positive $x$-axis and are uniformly distributed around $s$.
These rays together describe  a {\it (discrete) wavefront} initiated at $s$.
Let $r'$ and $r''$ be two rays in $\calR(s)$. Two rays are termed {\it successive} when
they are adjacent in the ordering of rays around any point from which those rays are initiated.
The number of rays in $\calR(s)$ is defined by the angle $\delta$ between successive rays; by expressing the value of $\delta$ in terms of $\epsilon$ later (in Section~\ref{sect:boundrays}), our algorithm ensures an $(1+\epsilon)$-approximation.
%\todo{if you add query, adjust the following}
When the wavefront strikes any vertex $v$, an ordered set $\calR(v)$ of rays are initiated from $v$, unless vertex $v$ is the destination $t$ itself.

\subsection{Types of rays}

We next discuss the types of rays that form part of our wavefront.
Only a subset of the  rays that have been initiated are  used in the propagation of the wavefront; a ray that is considered for propagation is said to have been {\it traced}.
We propagate the rays as described below:
let $e=(v',v'')$ be a common edge between faces $f'$ and $f''$.
Consider a ray $r \in \calR(v)$ that is traced along face $f'$ and suppose it strikes $e$ at a point $y' \in e$.

If the angle of incidence of $r$ onto $e$ is less than $\theta_c(f', f'')$, then $r$ refracts onto face $f''$ with the angle of refraction defined by Snell's law of refraction.
The case in which the angle of incidence of $r$ onto $e$ is greater than $\theta_c(f', f'')$ occurs only when $w_{f'} > w_{f''}$.
In this case, we do not propagate $r$ further as $r$ would not be part of a weighted shortest path to $t$.
%\todo{but it may cause an apprx weighted shortest path?}

	\begin{wrapfigure}{r}{0.5\textwidth}
	\centering
	\begin{minipage}[t]{\linewidth}
	\centering
	\includegraphics[totalheight=0.8in]{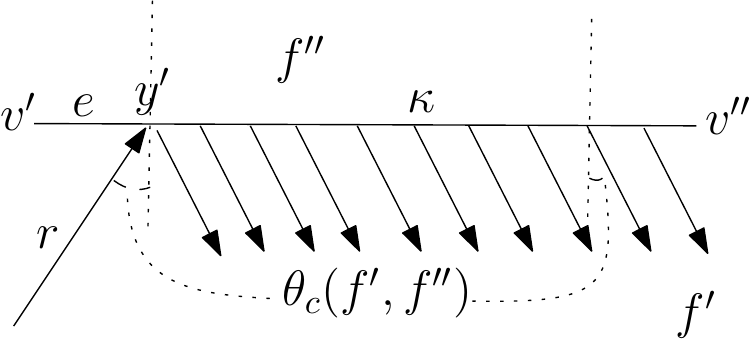}
	\vspace{-0.15in}
	\caption{\footnotesize Illustrating critical reflection}
	\label{fig:critreflequispaced}
	\end{minipage}
	\end{wrapfigure}

If the angle of incidence of $r$ onto $e$ is equal to $\theta_c(f', f'')$, then a weighted shortest path that uses the ray $r$ propagates along $e$ for a positive distance before critically reflecting back into face $f'$ itself. 
Let $\kappa = y'v''$ be the critical segment corresponding to this critical incidence of $r$ onto $e$ where $y'$ is the point of incidence of $r$ on $e$.
Since a weighted shortest path can be reflected back from any point on the critical segment $\kappa$ of $e$, 
our algorithm initiates rays from a discrete set of evenly spaced points on $\kappa$.
The number and position of points from which these rays are generated is again a function of $\epsilon$.
Let $\overrightarrow{v}$ be a vector normal to edge $e$, passing through point $y'$ to some point in face $f''$.
These rays critically reflect from $e$ back into face $f'$ while making an angle $-\theta_c(f', f'')$ with $\overrightarrow{v}$. 
(See Fig. \ref{fig:critreflequispaced}.)
We let $\calR(\kappa)$ be the ordered set of rays that originate from $\kappa$, and are ordered by the distance from $y'$. 
Two rays $r'$ and $r''$ in $\mathcal{R}(\kappa)$ are {\it successive} whenever there is no ray in $\mathcal{R}(\kappa)$ that occurs between $r'$ and $r''$ in the linear ordering along $\kappa$.

	\begin{wrapfigure}{r}{0.5\textwidth}
	\centering
	\begin{minipage}[t]{\linewidth}
	\centering
	\includegraphics[totalheight=1.3in]{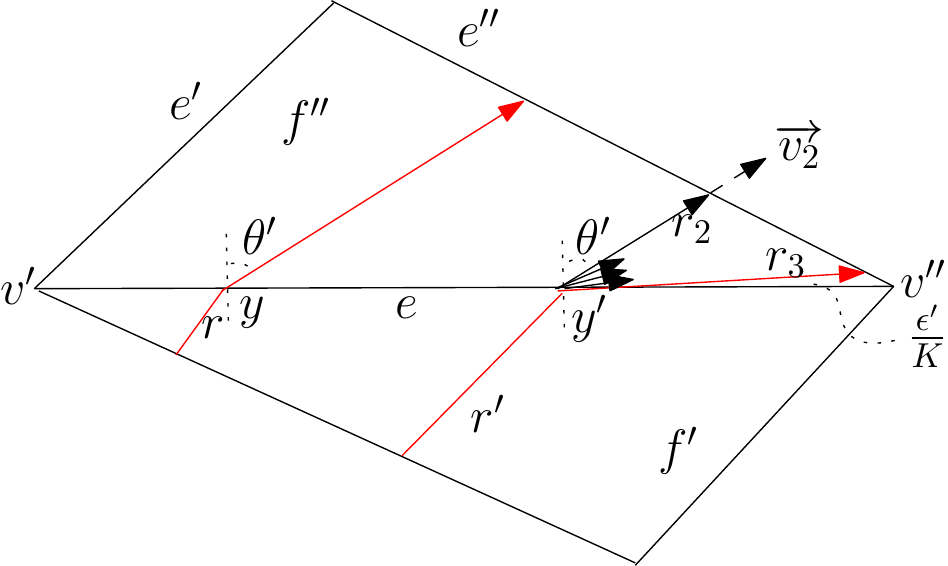}
	\caption{\footnotesize Illustrating rays initiated from a critical source $y'$ and the corresponding sibling pair $r, r_3$}
	\label{fig:initcritsrc}
	\end{minipage}
	\end{wrapfigure}

We next introduce an additional category of rays. To account for the divergence of rays and to ensure that
the wavefront is adequately represented, we further initiate an angle ordered set of {\it Steiner rays}, denoted by $\calR(y')$, from each critical point of entry $y'$ onto face $f''$. 
(See Fig. \ref{fig:initcritsrc}.)
Hence, a critical point of entry is also termed as a {\it critical source}.
These rays are motivated as follows: a pair of successive rays when traced along an edge sequence can diverge non-uniformly, with the divergence being large especially at angles close to critical angles.
To establish an approximation bound, instead of having a set of rays that is more dense from every source, we fill the gaps by generating Steiner rays from critical points of entries.
This helps in reducing the overall time complexity.
Let $v'$ and $v''$ be the endpoints of $e$.
For any source $u$, let $r, r'$ be successive rays in $\calR(u)$ such that $r$ is refracted and $r'$ is critically reflected.
Let $y$ (resp. $y'$) be the point at which $r$ (resp. $r'$) is incident to $e=(v', v'')$.
Let $\theta'$ be the angle of refraction of $r$.
Also, let $\overrightarrow{v_2}$ be a vector with origin at $y'$ and that makes an angle $\theta'$ with $\overrightarrow{v}$. 
As a result of discretization, there may not exist rays in $\calR(u)$ that intersect the cone $C(y',\overrightarrow{y'v''}, \overrightarrow{v_2})$.
Hence, to account for this region devoid of rays, an ordered set $\calR(y')$ of rays are initiated from $y'$ that lie in the cone $C(y',\overrightarrow{y'v''}, \overrightarrow{v_2})$, ordered by the angle each ray makes with respect to $\overrightarrow{v}$. 
Every ray in $\calR(y')$ is termed a {\it Steiner ray}.

To recapitulate, there are three sets  of rays that are used to define the discrete wavefront:
\begin{itemize}
\item \{$\calR(v)$ \hspace{0.005in} $|$ \hspace{0.005in} $v \in \calP$\}
\item \{$\calR(\kappa)$ \hspace{0.005in} $|$ \hspace{0.005in} $\kappa $ is a critical segment obtained  when a ray  $r \in \calR(v)$ strikes an edge $e$ at a critical angle\} 
\item \{$\calR(y')$ \hspace{0.005in} $|$ \hspace{0.005in} $y'$ is a critical point of entry for a ray $r \in \calR(v)$ for some $v \in \calP$\}
\end{itemize}
%We sometimes denote the origin of a ray $r$ as origin($r$); the origin($r$) could either be a vertex, a critical source, or a critical segment.

\subsection{Ray bundles}

We approximate the expansion of the continuous-Dijkstra wavefront by tracing rays in pairs (angle between the rays being less than $\pi$) so that each pair represents a section of the continuous-Dijkstra wavefront
that lies between the pair at corresponding distance $d$ along the chosen pair of rays.
Appropriately chosen  pairs suffice to represent the behavior of this section.
The rays could be initiated from either of these: (i) a vertex $v$, (ii) a critical source $c$, or (iii) a critical segment $\kappa$.
Let $o$ be one such source of rays.
We partition the set of rays with origin $o$ as follows: 
Let $B$ be a maximal set of successive rays in $\calR(o)$ such that all rays in $B$ cross the same edge sequence $\calE$ when traced from the source to  the current state of the discrete wavefront.
Then $B$ is said to be a {\it ray bundle} of $\calR(o)$.
Furthermore, $\calE$ is the edge sequence associated with $B$.
Let $r'$ and $r''$ be two rays in $B$ intersecting an edge $e \in {\cal E}$ such that the line segment defined by the points of intersection of $r'$ and $r''$ with $e$ intersects every other ray from $B$ (when traced).
Then the rays $r'$ and $r''$ are extremal rays of the bundle and are termed as the {\it sibling pair} of ray bundle $B$.
For any ray bundle $B$, instead of tracing all the rays in $B$, we trace only the sibling pair of $B$.
This helps in reducing the number of event points, hence the time complexity.

\begin{wrapfigure}{r}{0.5\textwidth}
\begin{minipage}[t]{\linewidth}
\vspace{-20pt}
\hspace{-10pt}
\begin{center}
\includegraphics[totalheight=1.0in]{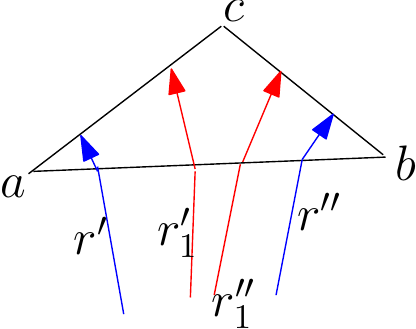}
\end{center}
\vspace{-20pt}
\caption{\footnotesize Split of a ray bundle }
\label{fig:bundsplit}
\vspace{-10pt}
\end{minipage}
\end{wrapfigure}

The ray bundles and corresponding sibling pairs are updated as the rays progress further.
For the sibling pair $r', r''$ of ray bundle $B$, when the sequence of edges intersected by $r'$ changes from the sequence of edges intersected by $r''$ for the first time, we need to split the bundle of rays into two and determine new sibling pairs.  
A binary search among the rays in $B$ is used to trace the appropriate rays across the edge sequence of $B$ and form new sibling pairs. 
Let $abc$ be a face $f$ in $\calP$.
(See Fig. \ref{fig:bundsplit}.)
Suppose rays $r'$ and $r''$ are siblings in a bundle $B$ before they strike edge $ab$ of $f$. 
However, when they are traced further, suppose $r'$ strikes edge $ac$ and $r''$ strikes edge $bc$.
At this instance of wavefront progression, the sequence of edges that are intersected by  $r'$ and $r''$ differ for the first time, and hence $r'$ and $r''$ do not belong to the same bundle from there on.
Using binary search over the rays in $B$, a pair of successive rays, say $r_1'$ and $r_1''$, are found to determine new sibling pairs: $r', r_1'$ and $r'', r_1''$.
Further, the ray bundle $B$ is {\it split} accordingly.
And, the wavefront is said to {\em strike} $c$. 
The details of the algorithm to find rays $r_1'$ and $r_1''$ is described later.

The sets of rays in $\calR(y')$ and $\calR(\kappa)$, where $y'$ is a critical point of entry and $\kappa$ is a critical segment, are handled similarly.
The rays in $\calR(\kappa)$ have the characteristic that the paths of any two rays in $\calR(\kappa)$ are parallel between any two successive edges in the edge sequence associated with the corresponding ray bundle.

\subsection{Tree of rays}

Let $v$ be a vertex in $\calP$. And let $\calR(v)$ be the set of rays initiated from $v$. 
A ray may lead to the initiation of critical sources as described before.
Let $S_1$ be the set of critical sources such that for any source $q$ in $S_1$, Steiner rays are initiated from $q$ when the discrete wavefront that  originates at $v$ strikes $q$.
Further, let $S_2, \ldots, S_{i-1}, S_i, \ldots, S_k$ be the sets of critical sources such that for every $q' \in S_i$ and $2 \le i \le k$, critical source $q'$ is initiated due to the critical incidence of a Steiner ray from a source in $S_{i-1}$.
%Formally, $S_i = \{ q'$ \hspace{0.02in} $|$ \hspace{0.01in} $q'$ is  a critical source with origin $q''$ for $q'' \in S_{i-1} \}$.
We organize the sources in the set ${\cal S} (v) =\{v\} \cup (\cup_j S_j)$ into a {\it tree of rays}, denoted by $\calT_R(v)$.
Each node of $\calT_R(v)$ corresponds to a source in ${\cal S}(v)$:
%comprises $\{v\} \cup (\cup_j S_j)$
more specifically, $v$ is the root node and every other node in $\calT_R(v)$ is a distinct critical source from $\cup_j S_j$. 
For any two nodes $u \in S_{i-1}, w \in S_i$ in $\calT_R(v)$, $u$ is the parent of $w$ in $\calT_R(v)$ if and only if the critical source $w $, located on an edge $e$, is initiated when a ray from $u$ strikes $e$ at the critical angle for edge $e$ (and the point of critical incidence is $w$). 
The in-order traversal of the tree $\calT_R(v)$ provides a natural order on the set of rays. 
A subset $S$ of rays are termed {\it successive} whenever rays in $S$ are a contiguous subsequence of the ordered set of rays

\subsection{Ray Bundles and Sibling Pairs}
In order to partition the rays in
$\calR(v)$, where $v$ is  a vertex, or in $\calR(\kappa)$, where $\kappa$ is a critical segment,
we generalize the definitions of  sibling pair and ray bundles.

We consider the rays in $\calR(v)$. 
Let ${\calE} = \{e_1, e_2, \ldots, e_k\}$ be an edge sequence.
For any $j \in [1, k]$, we say $e_j, \ldots, e_k$ is a {\it suffix of edge sequence $\calE$}. 
Two rays belonging to $\calT_{R}(v)$ are {\it siblings} whenever the edge sequence associated with one of them is a suffix of the edge sequence of the other. 
A maximal set $S$ of successive rays initiated from the nodes of $\calT_{R}(v)$ that are siblings to each other is a {\it ray bundle}.
For two rays $r', r''$ belonging to a ray bundle $B$ and points $p' = r' \cap e_k, p'' = r'' \cap e_k$, the rays $r'$ and $r''$ are termed as the {\it sibling pair of $B$} whenever the line segment $p'p''$ intersects every ray in $B$.
A similar definition holds for rays in $\calR(\kappa)$.
 
Therefore, there are two  kinds of sibling pairs possible:
Either (i) both the rays in a sibling pair originate from the same source or the
 sources of both the rays in a sibling pair belong to a tree of rays, ${\cal T}_R(v)$ or
(iii) both the rays in a sibling pair originate from the same critical segment.
A ray bundle of the first category is denoted as {\it a ray bundle of $\calT_{R}(v)$}. 
%Rays in a bundle are  generated from  a source vertex in $\calT_{R}(v)$ and  possibly its descendants.

Let $r', r''$ be a sibling pair of  rays in a bundle $B$ in $\calT_R(v)$.
Let $u, w \in \calT_R(v)$ be the respective origins of $r'$ and $r''$. 
Also, let $v'$ be the least common ancestor of $u$ and $w$ in $\calT_{R}(v)$.
We define $v'$ to be the {\em root} of $B$.
The path $P$ in tree $\calT_{R}(v)$ from $v'$ to $u$ comprising of a  sequence of critical sources is the {\it critical ancestor path} of $r'$ with respect to $r''$.  
Note that the critical ancestor path of a ray is always defined with respect to a sibling pair.
%Given that the rays in $\calT_R(v)$ do not intersect\todo{Hmmm...requires a proof},  
Given that rays in the bundle $B$ are successive rays,
nodes on the critical ancestor paths from $u$ and $w$ to $v'$
in the subtree $\calT_R(v)$ provide all the rays in the bundle $B$.

We next show a property of the rays in a bundle. We define two rays to be {\em pairwise divergent} during traversal of an edge sequence
${\cal E}$ if they do not intersect each other during the traversal. A bundle of rays is {\em divergent} during traversal of ${\cal E}$
if all rays in the bundle are pairwise divergent during the traversal.

\begin{lemma}
Let $B$ be a ray bundle of $\calT_{R}(v)$ and let ${\cal E}$ be the edge sequence associated with rays in $B$.
The rays in $B$ are pairwise divergent whey they traverse across the edge sequence ${\cal E}$.
\end{lemma}
\begin{proof}
Let $v'$ be the root of a bundle. The rays in $B$ that start from $v'$, termed $R(v')$, clearly diverge during the  traversal of  the sequence of edges ${\cal E}$, since divergence is preserved on refraction. Additionally, consider
rays in $B$ that are generated from a critical source, $w$, termed $R(w)$. Each pair of rays, a steiner ray $r$ originating at $w$ and a ray in $R(v')$, are pairwise divergent, by construction.
Similarly, every pair of ray in $R(w)$ is pairwise divergent.
These rays remain pairwise divergent  after refraction during propagation across edges in the sequence ${\cal E}$.
In general, let $R$ be  a bundle of
rays that is divergent during traversal of a subsequence of ${\cal E}$. These rays are  generated from a vertex and/or  a set of critical sources. Suppose the
rays in the bundle  strike an edge $e_j \in {\cal E}$ and refract, and also generate
a critical source $w' $ on $e_j$. Then $R \cup R(w')$, where $R(w')$ is the set of Steiner rays generated from $w'$,
is a set of rays with every pair of rays being pairwise divergent during traversal of ${\cal E}$ since  every  steiner ray in
$R(w')$ is pairwise divergent from each ray in $R$.
Thus these rays will not intersect each other as they strike edges $e_{j+1} \ldots e_k$.
\end{proof}

%From the construction of the rays at critical sources and the fact that all rays in a bundle intersect the
%same sequence of edges, it can be shown by induction that the rays do not intersect each other and 
The consequence of this Lemma is that  rays in a bundle $B$ can be angle ordered. 
As described later, this ordering and the critical ancestor path of a ray will be useful in efficiently splitting a sibling pair of $\calT_{R}(v)$.

%{\bf REQUIRES PROOF!!!!}

\section{Bounding the number of initiated and traced rays}
\label{sect:boundrays}

Before detailing the algorithm, we bound the number of initiated and traced rays in $\mathcal{P}$.
An approximate weighted shortest path from source $s$ to $t$ can be split into sub-paths such that each such sub-path goes between a pair of vertices. 
Further, each of these sub-paths are classified into following types:

\begin{itemize}
\item
a {\it Type-I path} from one vertex to another vertex that does not use any critical segment in between, and

\item
a {\it Type-II path} from one vertex to another vertex that critically reflects and uses at least one critical segment in-between.
\end{itemize}

We next bound the approximation achieved for both of these types of paths by establishing properties regarding the rays.
We show that given a Type-I (resp. Type-II) weighted shortest path $P$ from a vertex, say $u$, to another vertex, say $v$, there exists a Type-I (resp. Type-II) path $P'$ using only rays in our discretization such that $P'$ closely approximates $P$.
We first consider the divergence of successive rays that are within the same bundle. 
This divergence is due to refraction of rays. 
%Note, again, that the angles measures of the rays are with respect to a fixed co-ordinate system.

\begin{lemma}
\label{lem:depsilon}
Let $r', r''$ be two successive rays in a bundle in $\calT_R(u)$.
Also, let $q'$ and $q''$ be points on rays $r'$ and $r''$ respectively, at weighted Euclidean distance $d$ from $s$ and lying on the same face.
If the angle between rays $r'$ and $r''$ is upper bounded by $\frac{1}{2\mu}(\epsilon')^{n^2}\epsilon'$, then the weighted Euclidean distance of the line segment joining $q'$ and $q''$ is upper bounded by $d \epsilon'$. 
Here, $\epsilon'$ is a small constant (expressed in terms of input parameters).
\end{lemma}

\begin{wrapfigure}{r}{0.5\textwidth}
\centering
\begin{minipage}[b]{.4\textwidth}
\centering{\includegraphics[totalheight=1in]{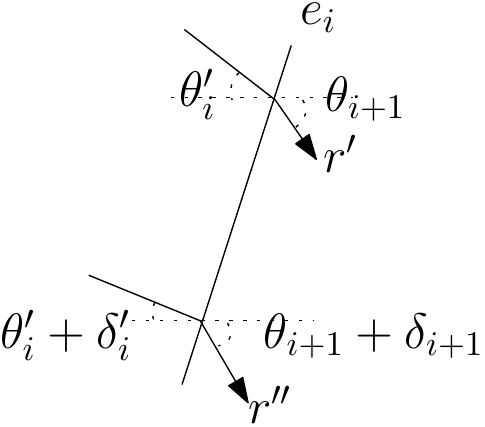}}
\caption{\footnotesize Illustrating the angle of refractions of rays $r'$ and $r''$}
\label{fig:refrangles}
\end{minipage}
\end{wrapfigure}

\begin{proof}
We first show a  bound on the divergence between a pair of successive rays  in $\calT_R(u)$.
Let the sibling pair of rays $r', r''$ in $\calT_R(u)$ traverse across the (same sequence of) faces $f_1, f_2 \ldots f_{p+1}$ with weights $w_1, w_2, \ldots, w_{p+1}$ respectively.
Let $\delta$ be the angle between $r'$ and $r''$.
Let $e_1, e_2, \ldots e_{p}$ be the edge sequence associated with this sibling pair, with $e_i$ being the edge
adjacent to faces $f_i$ and $f_{i+1}$ for every $i \leq p$.
Let $\theta_1', \theta_2', \ldots, \theta_{p}'$ be the angles at which the ray $r'$ is incident on edges $e_1, e_2, \ldots, e_{p}$ respectively.
(See Fig. \ref{fig:refrangles}.)
Let $\theta_2, \theta_3, \ldots, \theta_{p+1}$ be the angles at which the ray $r'$ refracts at edges $e_1, e_2, \ldots, e_{p}$ respectively. 
Similarly, let $\theta_1'+\delta_1', \theta_2'+\delta_2', \ldots, \theta_{p}'+\delta_{p}'$ be the angles at which the ray $r''$ is incident on edges $e_1, e_2, \ldots, e_{p}$ respectively.
And, let $\theta_2+\delta_2, \theta_3+\delta_3, \ldots, \theta_{p+1}+\delta_{p+1}$ be the angles at which the ray $r''$ refracts at edges $e_1, e_2, \ldots, e_{p}$ respectively. 
(See Fig. \ref{fig:refrthm}.)
For every $i$, if a critical angle exists for edge $e_i$, assume that both $\theta_i'$ and $\theta_i' + \delta_i'$ are less than that critical angle.
Further, for every $i$, assume that $\delta_i$ is a small positive number less than $\epsilon'$
(we will justify this assumption later). We will first provide a bound on $\delta_p$ in terms of $\delta_1$.  \hfil\break
\hfil\break
For every integer $i \in [1, p]$, using Snell's law,
\begin{eqnarray}
w_i\sin{\theta_i'}=w_{i+1}\sin{\theta_{i+1}} \label{eq:refr}
\end{eqnarray} 
and 
\begin{eqnarray}
{w_i\sin{(\theta_i'+\delta_i')}=w_{i+1}\sin{(\theta_{i+1}+\delta_{i+1})}}.  \nonumber
\end{eqnarray} 
Since $\delta'_i=\delta_i$, %by an argument from geometry,  for every integer $i \in [1, p-1]$,
\begin{eqnarray}
{w_i\sin{(\theta_i'+\delta_i)}=w_{i+1}\sin{(\theta_{i+1}+\delta_{i+1})}}.  \nonumber
\end{eqnarray} 
Both $\delta_i$ and $\delta_{i+1}$ are assumed to be very small; thus we approximate $\sin{\delta_i}, \cos{\delta_i}, \sin{\delta_{i+1}}$, and $\cos{\delta_{i+1}}$ with the first term of the series expansion (an analysis with higher order  terms reveals no additional benefit) and thus:
\begin{eqnarray}
\lefteqn{w_i\sin{\theta_i'}+w_i\delta_i\cos{\theta_i'} = w_{i+1}\sin{\theta_{i+1}}+w_{i+1}\delta_{i+1}\cos{\theta_{i+1}}} \nonumber \\
& \Rightarrow & w_i\delta_i\cos{\theta_i'}=w_{i+1}\delta_{i+1}\cos{\theta_{i+1}}  \hspace{0.08in} (from (\ref{eq:refr})) \nonumber \\
& \Rightarrow & \delta_{i+1}=\frac{w_i}{w_{i+1}}\frac{\cos{\theta_i'}}{\cos{\theta_{i+1}}}\delta_i  \label{eq:delta}
\end{eqnarray}
Letting $\max_{1\leq i \leq p} \frac{\cos{\theta_i'}}{\cos{\theta_{i+1}}}
 = \beta$, the above leads to  
\begin{eqnarray}
\delta_{p+1} \le \frac{w_1}{w_{p+1}}\beta^p\delta_1  ,  \ \  p \geq 1\label{eq:deltan}
\end{eqnarray}

Let $\theta^i_c$ be the critical angle corresponding to faces $f_i$ and $f_{i+1}$.
If $\theta'_i $ is close to the critical angle $\theta^c_i$ then $\beta$ grows unbounded. 
We thus restrict $\theta'_i$ to $\theta^i_c - K\epsilon'$, where $K$ is a constant defined as $\max_i(\frac{w_{i+1}}{w_i \sin \theta^i_c})$. 
And, we note that $K$ is upper bounded by $\mu^2$.
This ensures that every refracted ray is at an angle less than $\frac{\pi}{2} - \epsilon'$ with respect to the normal at edge $e_i$.
Thus $\beta \leq \frac{1}{\cos(\frac{\pi}{2} - \epsilon')} \leq \frac{1}{\epsilon'}$.
The error introduced due to this assumption will be bounded later.

\begin{figure}[h]
\begin{minipage}[t]{\linewidth}
\begin{center}
\includegraphics[totalheight=1.6in]{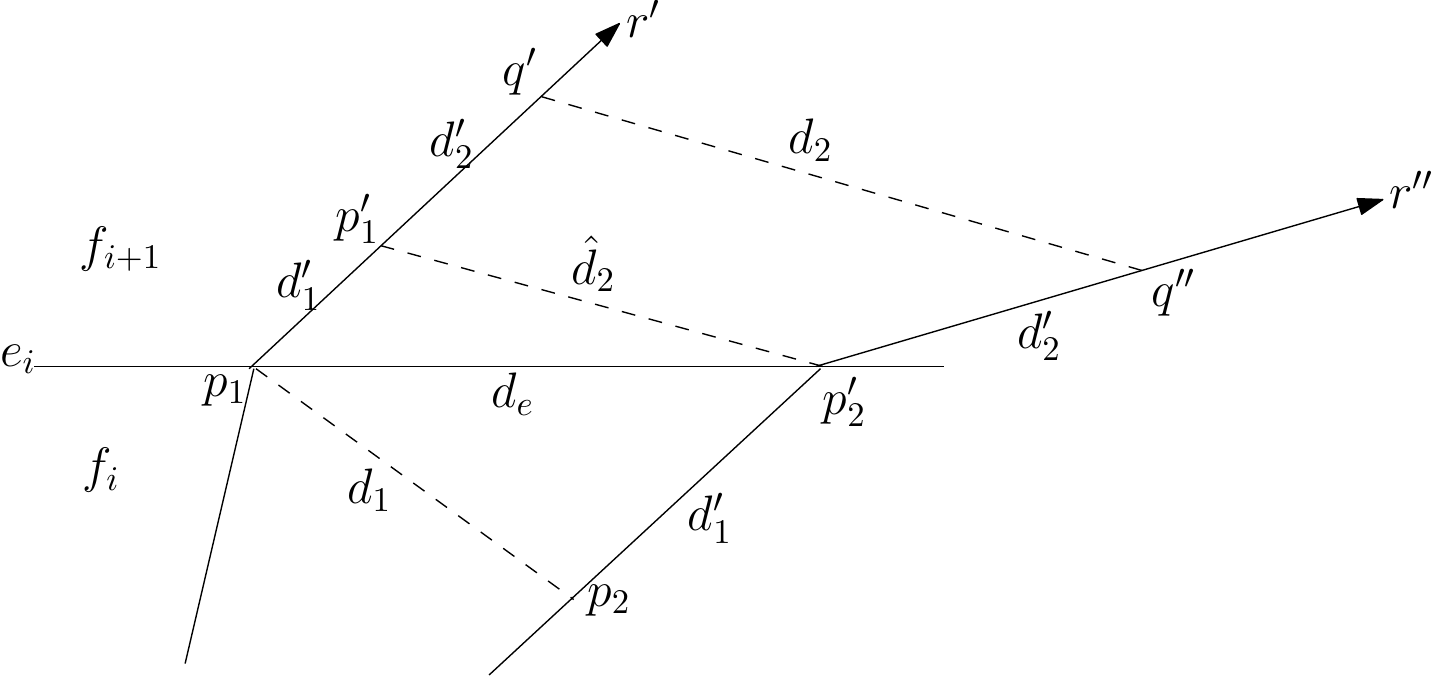}
\caption{\footnotesize
Illustrating the construction in proving Lemma \ref{lem:depsilon}
\scriptsize{}
(Except for $d_e$, all the distances' shown in the figure are weighted Euclidean distances.)
\normalsize{}
}
\label{fig:refrthm}
\end{center}
\end{minipage}
\vspace*{-0.1in}
\end{figure}

To compute the weighted Euclidean distance between two points at weighted Eulcidean distance $d$ from the source, first we restrict our attention to two points that are on successive rays in the same bundle at weighted Eulcidean distance $d$ from the source and that lie on the same face. 
Let the two points, $q'$ and $q''$, both on the same face $f_{i+1}$, be such that they
(i) lie on a pair of successive rays $r'$ and $r''$ in the same bundle in $\calT_R(u)$, and (ii) are at equal weighted Euclidean distance from  $u$.
Let $p_1$ be the point of incidence of $r'$ on an edge $e_i$ (which is common to $f_i$ and $f_{i+1}$) with an angle of incidence $\theta_i'$ and let it refract from $p_1$ at an angle of refraction $\theta_{i+1}$. 
Further, let the ray $r''$ be incident on $e_i$ at $p_2'$ with an angle of incidence $\theta_i'+\delta_i'$ and refract from $p_2'$ at an angle of refraction $\theta_{i+1}+\delta_{i+1}$. 
(See Fig. \ref{fig:refrthm}.)
W.l.o.g., assume that the weighted Euclidean distance from vertex $u$ to $p_2'$ is larger than the distance from $u$ to $p_1$. 
Consider an additional point $p_2$ located on ray $r''$ such that both $p_1$ and $p_2$ are in the face $f_i$, and the weighted Euclidean distance from $u$ to either of these points is $d'$. 
Let $d_1$ be the weighted Euclidean distance between $p_1$ and $p_2$.
Also, points $q', q''$ are located on rays $r', r''$ respectively such that $q', q''$ are in the region $f_{i+1}$ with weight $w_{i+1}$.
Further, we let $d$ be the weighted Euclidean distance from $u$ to either of these points and let $d_2$ be the weighted Euclidean distance between $q'$ and $q''$.

We establish a bound on $d_2$ for $i \geq 0$.
For every $j$, we let $\kappa_j =  \frac{w_1}{w_j} \beta^{j-1} \delta_1$.
We first show that if $d_1 \le 2d' \kappa_i$ then $ d_2 \le 2d \kappa_{i+1}$. 
Clearly, $d_1 \leq 2d' \kappa_1$, i.e., when $i=1$.

We first consider the case when points $q'$ and $q''$ are on the same face $f_{i+1}$.
Let $p_1'$ be the point that is at weighted Euclidean distance $d_1'$ from $p_1$.
Let $d_e$ be the Euclidean distance between $p_1$ and $p_2'$ along edge $e$.
Also, let $\hat{d}_2$ be the weighted Euclidean distance between $p_1'$ and $p_2'$.

Now by triangle inequality
\[  d_e \leq \frac{1}{w_i} ( d_1 +d_1') \]
By assumption $d_1 \leq 2d' \kappa_i$ and, by triangle inequality
\begin{eqnarray*}
 \hat{d}_2  &\leq & w_{i+1} (\frac{d_1'}{w_{i+1}} +  \frac{1}{w_i} ( d_1 +d_1')) \\
 &\leq&  d_1' + \frac{w_{i+1}}{w_i}(2d' \kappa_i + d_1') \\
 & \leq& 2(d'+d_1') \kappa_{i+1}
\end{eqnarray*}
The last inequality follows from the assumption that $\kappa_{i+1} \geq \kappa_i \geq 1$.

Next consider the weighted Eulcidean distance between $q'$ and $q''$, where $q'$ is at weighted Euclidean distance $d_2'$ from $p_1'$, and $q''$ is at weighted Euclidean distance $d_2'$ from $p_2'$.
\[ d_2 \leq \hat{d}_2 + \delta_{i+1}d_2' \]
using the small angle approximation of distance along a circular arc subtending a small angle.
Thus, since $\delta_{i+1} \leq \kappa_{i+1} $,
\[ d_2 \leq  2(d'+  d_1'+ d_2') \kappa_{i+1}  = 2d\kappa_{i+1}\]
Since $\kappa_j =  \frac{w_1}{w_j} \beta^{j-1} \delta_1$, $\beta \leq \frac{1}{\epsilon'}$ and $j \leq n^2$, if we choose $\delta_1 \leq \frac{(\epsilon')^{n^2}\epsilon'}{2\mu}$ then the weighted Euclidean distance $d_2$ is bounded by $ d\epsilon'$.
\end{proof}

The following corollary to the above eliminates the restriction for points $q'$ and $q''$ being on the same face.
\begin{cor}
\label{lem:depsilonC}
Let $r', r''$ be two successive rays in a bundle $B$ in $\calT_R(u)$.
Also, let $q'$ and $q''$ be two points on $r'$ and $r''$ respectively, at weighted Euclidean distance $d$ from $s$.
If the angle between rays $r'$ and $r''$ is upper bounded by $\frac{1}{2\mu}(\epsilon')^{n^2}\epsilon'$, then the weighted Euclidean distance of the line segment joining $q'$ and $q''$ is upper bounded by $d \epsilon'$. 
Here, $\epsilon'$ is a constant (expressed in terms of input parameters).
\end{cor}
\begin{proof}
Let $\delta$ be the angle between $r'$ and $r''$.
Consider the locus of points at weighted Euclidean distance $d$ from the root of $\calT_R(u)$ that lie in the region bounded by rays $r'$ and $r''$.
Let $q'=p_0, p_1, p_2, \ldots p_k, q''=p_{k+1}$ be the sequence $\cal{S}$ of points of intersections of this locus with the edges of $\calP$.
Consider the rays in the bundle $B$ that strike points in $\cal{S}$.
Further, for $0 \le i \le k$, let $\delta_i$ be the angle between rays that strike points $p_i$ and $p_{i+1}$ of $\cal{S}$. 
Then the Lemma~\ref{lem:depsilon} shows that the weighted Euclidean distance between $p_i$ and $p_{i+1}$ is $\frac{2d\mu\delta_i}{(\epsilon')^{n^2}}$ for every $0 \le i \le k$.
The overall bound follows by summing up the weighted Euclidean distances' between successive points in $\cal{S}$.
\end{proof}

Before we proceed further, for simplicity, we establish another restriction on the angle that two successive rays in $\calT_R(u)$ make at their source $u$. To ensure that each triangle in our planar subdivision gets intersected by at least two rays, we further restrict angle between any two successive rays so that $d \epsilon' \leq \frac{l_{min}}{4}$.
Here, $d= nw_{max}l_{max}$ is the largest weighted Euclidean distance in $\calP$; $l_{max}$ and $l_{min}$ are respectively the lengths of edges with maximum and minimum Euclidean lengths in ${\cal P}$.
Thus this establishes another condition on $\epsilon'$, i.e.
\begin{equation}
\epsilon' \leq  \frac{l_{min}}{4n w_{max} l_{max}}  \label{eqnepsilon}
\end{equation}
We can now prove the following.

\begin{lemma}
\label{lem:errorX}
Let $u$ be a vertex and let $v$ be a point on an edge $e$ in $\cal{P}$ such that the weighted Euclidean distance from $u$ to $v$ via a Type-1 ray in  bundle $B$ is $d(u,v)$.
If $v$ is not a vertex then there exists at least one traced ray in $B$ that is incident onto a point $v'$ belonging to edge $e$ such that $d(u,v') \leq d(u,v)(1+ \eta \epsilon')$. 
Here, $\eta= 2(1+ \frac{1}{\cos{\theta_{cm}}})$ and $\theta_{cm}$ is the maximum value of the critical angle at any edge in $\cal{P}$. 
\end{lemma}

\begin{proof}
Consider the Type-1 ray $r$ from $u$ to $v$, which is not traced and a ray $r'$ neighboring $r$ in the angular order of rays originated at $u$ is traced by the algorithm. 
(See Fig. \ref{fig:refrthmc}.)

\begin{wrapfigure}{r}{0.55\textwidth}
\centering
\includegraphics[totalheight=0.7in]{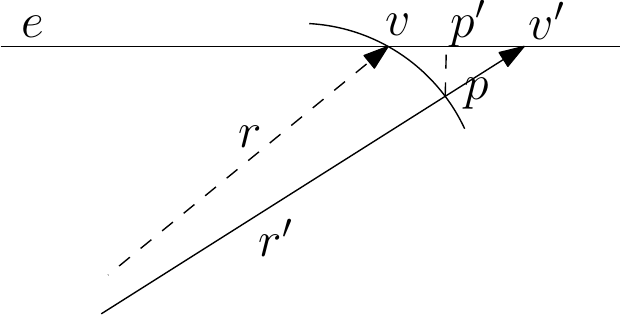}
\caption{\footnotesize Illustrating the construction in proving the Lemma \ref{lem:errorX} }
\label{fig:refrthmc}
\end{wrapfigure}

\ignore {
\begin{center}
\includegraphics[totalheight=1.0in]{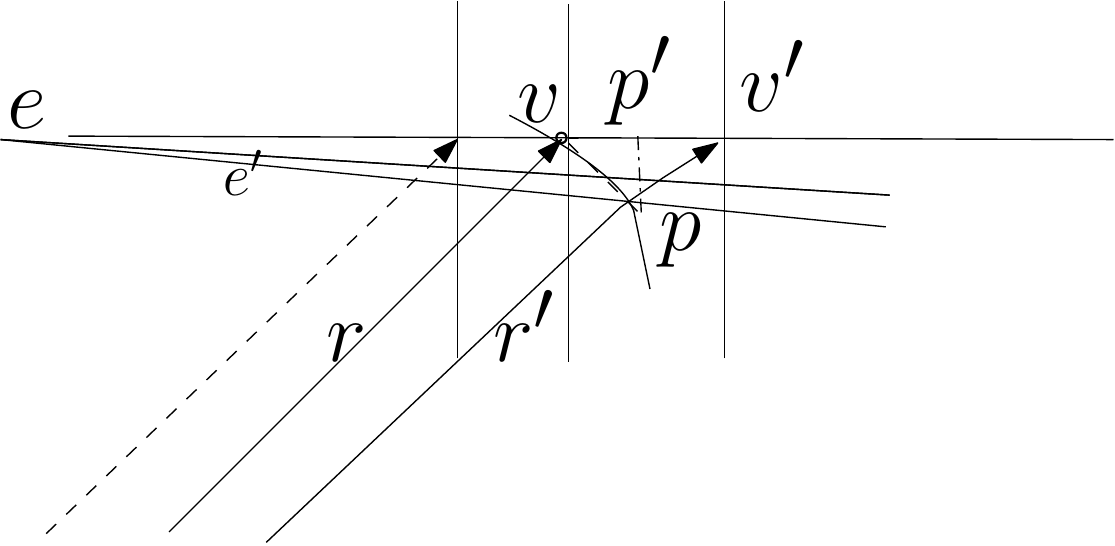}
}

Recall that the angle between any two successive rays is $\delta$. 
We consider the weighted Euclidean distance of points on the segment $(v,v')$.
If $d(u,v') \leq d(u,v)$ then we are done.
If not, then  consider the point $p$ on the ray $ r'$,  at distance $d(u,v)$ from $u$.
First, assume that $p$ and $v$ are on a common face. 
Let $w_1$ be the weight of the common face.
Thus using Lemma~\ref{lem:depsilon},  $d(p,v) \leq d(u,v) \epsilon'$, since the two rays $r$ and $r'$ belong to the same bundle. Moreover,
the Euclidean distance between $p$ and $v$ is upper bounded by $d(u,v)\frac{\epsilon}{w_1}$.
Note that the angle of incidence $\theta$ is less than or equal to the critical angle of $e$. 
Hence, the $\widehat{vv'p}$ angle is greater than $\frac{\pi}{2}-\theta_{cm}$.
Let $p'$ be the projection of $p$ onto $\overline{vv'}$.
Then $\Vert pp' \Vert = \Vert vp \Vert \sin{\widehat{pvv'}} = \Vert pv' \Vert \sin \theta'$.
Thus, $\Vert pv' \Vert \leq \frac{\Vert vp \Vert}{\cos{\theta_{cm}} }$
and 
$d(v,v') \leq d(u,v) \epsilon'(1+\frac{1}{\cos{\theta_{cm}}} )$.

Next consider the case in which $p$ and $v$ lie on different faces. 
Let the line segment $vp$ be crossing faces $f_1, f_2, \ldots, f_k$.
(Here, $f_1$ is the face containing $v$.)
We assume w.l.o.g. that $w_1 \leq w_2 \ldots \leq w_k$, so that rays $r$ and $r'$ diverge as they progress across these faces.
By Corollary~\ref{lem:depsilonC}, the weighted Euclidean distance between $p$ and $v$ is upper bounded by $d(u,v)\epsilon'$;
and the Euclidean distance between $p$ and $v$ is bounded by $\frac{d(u,v)\epsilon'}{w_1}$. 
Furthermore, the line segment $vv'$ is incident to $e$ at an angle larger than the critical angle of $e$. 
By similar arguments as above, the bound follows.
\end{proof}

\begin{cor}
\label{cor:errorX}
Let $u$ be a vertex and let $v$ be a point on an edge $e$ in $\cal{P}$ such that the weighted Euclidean distance from $u$ to $v$ via a Type-1 ray in  bundle $B$ is $d(u,v)$.
If $v$ is a vertex then there exists at least one traced ray in $B$ that is incident onto a point $v'$ belonging to an edge $e$ incident to $v$ such that $d(u,v') \leq d(u,v)(1+\eta \epsilon')$.
Here, $\eta= 2(1+ \frac{1}{\cos{\theta_{cm}}})$ and $\theta_{cm}$ is the maximum value of the critical angle at any edge in $\cal{P}$. 
\end{cor}

The above Lemma and Corollary are used to bound the approximation of the distance from $u$ to a point with the condition that the point lies between any two successive rays in the same  bundle $B$.
A Type-I weighted shortest path from $u$ to an arbitrary vertex $v$, however, could use a critical source and hence use multiple bundles.
The following theorem consider these types of paths in ensuring an $(1+\epsilon)$-approximation as well as the time complexity. 

\begin{lemma}
\label{lem:refrnoncritical}
If the angle between a pair of successive rays is as specified in Lemma~\ref{lem:depsilon}, then a Type-I weighted shortest path from any vertex $u$ to another vertex $v$ in $\calP$ can be approximated to within a factor of $(1+ 2n^2 \eta \epsilon')$ using rays in $\calT_R(u)$, where $\eta= 2(1+ \frac{1}{\cos{\theta_{cm}}})$.
\end{lemma}
\begin{proof}
We first consider when an optimum weighted shortest path from $u$ to $v$ traverses a point $p_1'$ such that a weighted shortest path to $p_1'$ using the rays in ${\calT}_R(u)$ has a critical source $p_1$ in-between.
Consider two such successive rays $r_i \in \calT_R(u)$ and $r_j \in \calT_R(u)$ that originate from $u$ and that lie
in the same bundle $B$ (refer Fig. \ref{fig:refrthmb}) such that
 $r_i $ is incident to a point $p_1''$ located on edge $e$, at an angle less than the critical angle of $e$
and $r_j \in \calT_R(u)$  is incident to a point $p_1$ located on edge $e$, at an angle equal to the critical angle of $e$.  
Due to inequality (\ref{eqnepsilon}), the condition imposed on $\epsilon'$ in the assumption of this theorem, and the construction of the rays, such rays are guaranteed to exist in the same bundle.

Let $d_{P}(x, y)$ represent the distance between points $x$ and $y$ along path $P$. 
Suppose an optimal weighted shortest path $P_{opt}$ intersects edge $e$ between $p_1''$ and $p_1$ at $p_1'$.
Let $P_{opt}$ refract at point $p_1'$ with $\theta_{opt}$ as the angle of refraction. 
Consider a path $P$ that approximates $P_{opt}$: it uses ray $r_j$ from $v_1$ to $p_1$ and then uses a weighted shortest path from $p_1$ to some point $v_2$. %using the rays from $p_1$ in $\mathcal{P}$.
Then the modified path has length specified by $d_{P}(u, v) \leq d_{P}(u, p_1)+d_{P}(p_1, v)$ where $d_{P}(p_1,v)$ is a weighted shortest path from $p_1$ to $v$ in the discretized space.
By Lemma~\ref{lem:errorX},
$d_{P}(u, p_1) \le d_{opt}(u, p_1') (1+ \eta \epsilon' )$
with $d_P(p_1',p_1) \leq \eta \epsilon'$.
Further, $d_{opt}(p_1, v) \le d_{P}(p_1,p_1') + d_{opt}(p_1',v)$.
Applying  Lemma~\ref{lem:errorX} again, 
$d_{P}(p_1,v) \leq d_{opt}(p_1,v)(1+ \eta \epsilon')$ .
Thus $d_{P}(u,v) \leq  d_{opt}(v_1,v_2)(1+ 2\eta \epsilon')$.
Since the cardinality of any edge sequence of a path is  $O(n^2)$,
repeating the above analysis for all edges that the ray might encounter,
results in an error factor of $O(n^2 \eta \epsilon') $.
The case when no critical source is in-between is handled by  Lemma~\ref{lem:errorX}.
This completes the proof of Lemma~\ref{lem:refrnoncritical}.
\begin{figure}[h]
\begin{minipage}[t]{\linewidth}
\begin{center}
\includegraphics[totalheight=0.8in]{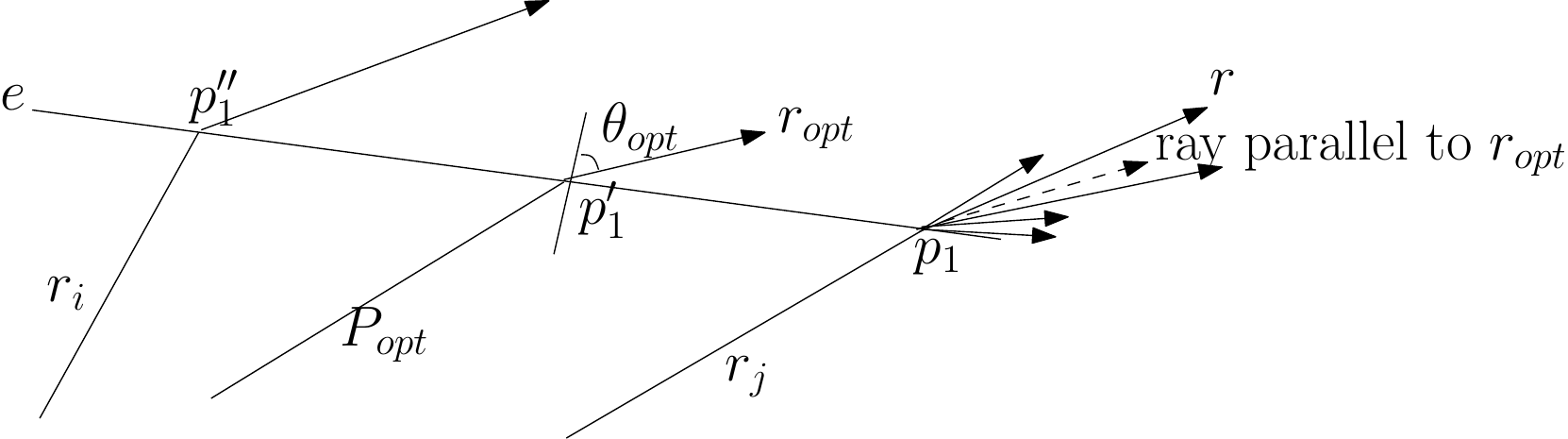}
\caption{\footnotesize Illustrating the construction in proving Lemma~\ref{lem:refrcritical} }
\label{fig:refrthmb}
\end{center}
\end{minipage}
\vspace*{-0.1in}
\end{figure}
\end{proof}

We next consider Type-II paths.
\begin{lemma}
\label{lem:refrcritical}
If the angle between any pair of successive rays is as specified in Lemma~\ref{lem:depsilon} and rays generated from a critical segment are $\epsilon'$ Euclidean distance apart, then a Type-II weighted shortest path from any vertex $u$ to another vertex $v$ in $\calP$ can be approximated to within a factor of $(1+ 2n^2 \eta \epsilon')$ using rays in $\calT_R(u)$, where $\eta= 2(1+ \frac{1}{\cos{\theta_{cm}}})$.
\end{lemma}
\begin{proof}
An optimal Type-II path $P$ can be partitioned into two: a path $P_1$ from $v$ to $\kappa$ and a path $P_2$ from $\kappa$ to $w$.
Let the critical segment $\kappa$ have as a critical point of entry, $y$ on edge $e$.
We first show that a good approximation to the path $P_2$ from $\kappa$ to $w$ can be found.
Let $P$ be a weighted shortest path that originates at a point $p$ on $\kappa$ and strikes $w$.
This path get critically reflected by edge $e$; hence, exits $\kappa$.
Since rays are generated $\epsilon'$ apart, there exist a sibling pair of rays that originate from two
points $p'$ and $p''$ adjacent to $p$ and strike edges incident to $w$.
These rays are parallel to $P$ and a proof similar to Lemma~\ref{lem:errorX} shows that 
$P$ is approximated to well within a factor $(1+\eta \epsilon')$.
Furthermore, the points $p'$ and $p''$ can be discovered by a binary search for the point $p$, with an additive error of $\epsilon'$.
Finally, using Lemma~\ref{lem:refrnoncritical}, it is clear that the weighted shortest distance from $v$ to $y$ is approximated to within a factor of $(1+2n^2\eta \epsilon')$ where $d_1$ is the optimal weighted Euclidean distance from $v$ to $\kappa$. 
\end{proof}

We finally finish the analysis with the observation that an approximate weighted shortest path from $s$ to $t$ can be split into at most $n$ approximate weighted shortest paths.
Thus we have the following Theorem to summarize.
\begin{theorem}
Let ${\cal P}$ be a weighted triangulated polygonal domain with vertex set $V $. %\cup \{s,t \}$.
Let $\epsilon' = \min \{ \frac{\epsilon}{n^3\mu\eta},
\frac{l_{min}}{4n w_{max} l_{max}} \}$
where %$\mu = \frac{w_{max}}{w_{min}}$,
$\eta= (1+ \frac{1}{\cos ( \theta_{cm})})$
and $\theta_{cm}$ is the maximum critical angle of any edge in $\cal{P}$.
%$l_{max}$ (resp. $l_{min}$) is the length of edge with maximum (resp. minimum) Euclidean length in $\cal{P}$, and $w_{max}$ (resp. $w_{min}$) is the weight of a maximum (resp. minimum) weighted face in ${\cal P}$. 
A weighted shortest path from $s \in V$ to $t \in V $ in ${\cal P}$ can be approximated to within a factor of $(1+\epsilon)$ using rays in $\cup_{u  \in V} \calT_R(u)$.
\end{theorem}

\section{Details of the algorithm}
\label{sect:algodetails}

The algorithm is event-driven, where the events are considered by their weighted Euclidean distance from $s$.
The algorithm starts with initiating a set $\calR(s)$ of rays that are uniformly distributed around $s$.
With each traced ray $r$, we save the point of its origin, points of refraction, critical points of entry and critical points of exit along the path traced by $r$  i.e., as these event points occur in order.
As $r$ is traced, we append the event points to the list associated with $r$.
The various types of events that need to be both determined and handled are described in the following Subsections.

\subsection{Initiating rays from a vertex}
\label{subsect:initraysvert}

This procedure is invoked to initiate rays from a given vertex, say $v$, when the discrete wavefront strikes $v$.
Since $s$ is also a vertex of $\calP$, this procedure is also used in initiating rays from $s$ as well. 

Let $f'$ be the face along which a ray has been determined to strike a vertex $v$ on the face.
Let $\calF(v)$ be the collection of all the faces incident to $v$ except for $f'$.
The set $\calR(v)$ of rays are initiated from $v$, and all these rays lie on the set $\calF(v)$ of faces.
Due to Proposition~\ref{prop:noncrossing} (non-crossing property of weighted shortest paths), we do not initiate rays from $v$ over the face $f'$. 
Further, the angle between any two successive rays in $\calR(v)$ at $v$  are bounded by $\delta$ (whose value is bounded as described in Section \ref{sect:boundrays}).

\begin{wrapfigure}{r}{0.5\textwidth}
\begin{minipage}[t]{\linewidth}
\begin{center}
\includegraphics[totalheight=0.9in]{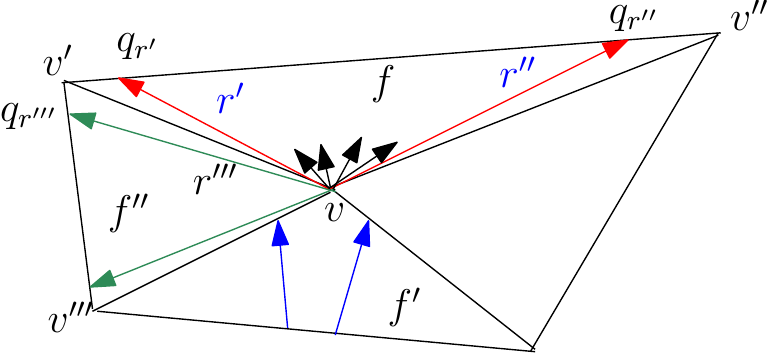}
\caption{\footnotesize Illustrating successive rays (shown in blue color) striking $v$ from $f'$ and a ray bundle being initiated on a face $f \in \calF(v)$; $r', r''$ is the sibling pair of $B_f$.  When $r'$ along $f$ and/or $r'''$ along $f''$ strike $v'v''$ and $v'v'''$ edges respectively, a discrete wavefront is initiated from $v'$.}
\label{fig:initverta}
\end{center}
\end{minipage}
\end{wrapfigure}

Consider any face $f \in \calF(v)$.
Let $vv', vv'', v'v''$ be the edges bounding $f$. 
(See Fig. \ref{fig:initverta}.)
The set $B_f \subseteq \calR(v)$ of rays that lie on face $f$ is a ray bundle:
every ray in $B_f$ strikes $v'v''$ before striking any other edge in $\calP$.
For every such ray bundle $B_f$, we find a sibling pair corresponding to $B_f$.
Let $r'$ (resp. $r''$) be the ray in $B_f$ that strikes $v'v''$ at $q_{r'}$(resp. $q_{r''})$ such that there does not exist a ray in $B_f$ that strikes $v'v''$ between $q_{r'}$ (resp. $q_{r''}$) and $v'$ (resp. $v''$).
We do binary search (with respect to edges $vv'$ and $vv''$) among the rays in $\calR(v)$ to find both the rays $r'$ and $r''$. We trace the sibling pair $r'$ and $r''$ as long as both the rays refract along the same sequence of edges $\mathcal{E}$ until the bundle is forced to split.
Let $e$ be the last edge in the edge sequence $\mathcal{E}$ at this stage.
Let $q'$ and $q''$ be the points on edge $e$ to which rays $r'$ and $r''$ are incident when they are traced.
Let $d_v$ be the weighted Euclidean distance from $s$ to $v$ i.e., when the discrete wavefront struck $v$.
The following sets of event points are pushed to the event heap:
the event point corresponding to tracing the ray $r'$ (resp. $r''$) to $q'$ (resp. $q''$) that occurs at distance equal to the weighted Euclidean distance $d_v$ added with the weighted distance along $r'$ (resp. $r''$) from $s$.
Further, the corresponding sibling pair is saved with each event point.
The splitting of sibling pairs is further detailed in Subsection \ref{subsect:splitsibpair}.

At the initialization step, let $f, f''$ be two adjacent faces in $\calF(v)$.
(See Fig. \ref{fig:initverta}.)
Let $vv'v''$ (resp. $vv'v'''$) be the triangle defining $f$ (resp. $f''$).
Let $r'$ be the ray in $\calR(v)$ that lies on $f$ and let $r'''$ be the ray in $\calR(v)$ that lies on face $f''$ such that no ray in $\calR(v)$ lies between $r'$ and $r'''$.
Also, let $q_{r'}$ (resp. $q_{r'''}$) be the point on edge $v'v''$ (resp. $v'v'''$) to which ray $r'$ (resp. $r'''$) is incident when traced.
Then the event point for initiating a discrete wavefront from $v'$ is pushed to the event heap with the key value $\min(d_v + w_{f} \Vert vq_{r'} \Vert + w_{v'v''} \Vert q_{r'}v' \Vert, d_v + w_{f''} \Vert vq_{r'''} \Vert + w_{v'v'''} \Vert q_{r'''}v' \Vert)$.

To improve the time complexity, we use the non-crossing property of (weighted) shortest paths (Proposition~\ref{prop:noncrossing}): whenever a ray bundle $B$ that originates at a vertex $v'$ strikes another vertex $v$, we save that information with $v$ so that whenever another ray bundle $B'$ that has originated from the same vertex $v'$ splits at vertex $v$ at an event corresponding to a larger distance, we do not split $B'$ in order to propagate from $v$.
To achieve this, whenever a vertex $v$ is struck by a ray bundle $B$, we save the origin of $B$ with $v$ in a set.
Further, whenever a different ray bundle $B'$ splits due to vertex $v$, we check this set and (i) update the bundle that first strikes $v$, say $B''$ (ii) eliminate progressing any bundle (including $B'$) whose origin is same as $B''$ if it crosses the bundle $B''$.

\subsection{Handling the ray striking an edge critically}
\label{subsect:crit}

Let $r_1, r_2$ be a sibling pair and let $e(v', v'')$ be a common edge to faces $f$ and $f'$.
Consider the following event point: A ray $r_1$ is traced along face $f$ and is critically incident to $e$ at point $y \in e$.
When this event occurs, we initiate two kinds of rays:
set $\calR(\kappa)$ of rays that originate from the critical segment $yv''$ (denoted by $\kappa$);
set $\calR(y)$ of Steiner rays that originate from the critical source $y$. 
In the following Subsections, we describe algorithms to both initiate these sets of rays and to set up ray bundles. 

\subsubsection{Initiating rays from a critical segment}
\label{subsect:initrayscritseg}

\begin{figure}[h]
\begin{minipage}[t]{0.49\linewidth}
\begin{center}
\includegraphics[totalheight=1.2in]{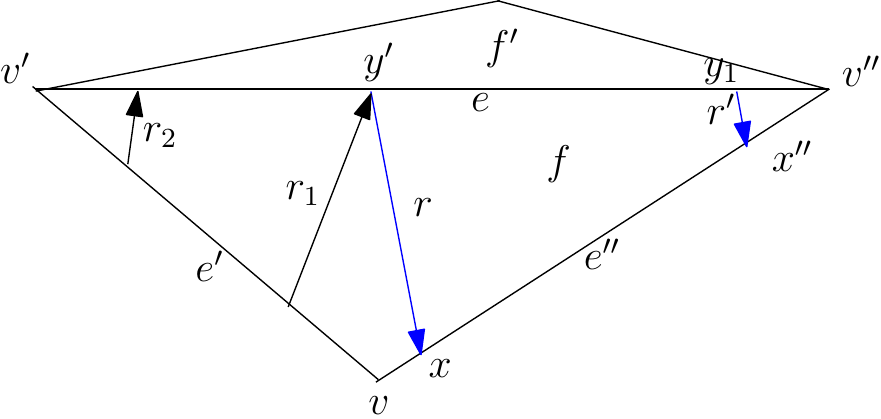}
\caption{\footnotesize Illustrating a sibling pair $r, r'$ originated at a critical segment that does not require splitting.} 
\label{fig:initcritseg}
\end{center}
\end{minipage}
\hspace*{0.04in}
\begin{minipage}[t]{0.49\linewidth}
\begin{center}
\includegraphics[totalheight=1.2in]{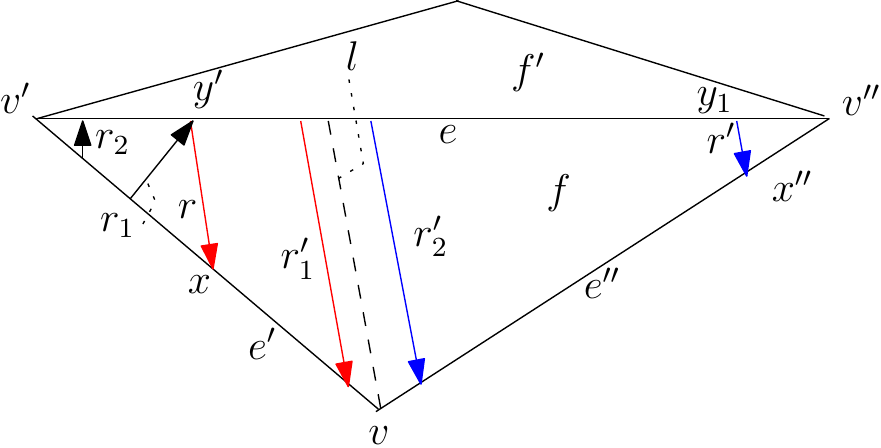}
\caption{\footnotesize Illustrating a sibling pair $r, r'$ originated at a critical segment that does require splitting.} 
\label{fig:splitsib-critseg}
\end{center}
\end{minipage}
\vspace*{-0.1in}
\end{figure}

We describe the procedure to initiate rays from the critical segment $\kappa$ first.
Given that a geodesic shortest path can be reflected back onto face $f$ from any point on $\kappa$, we initiate a discrete wavefront from $\kappa$.
Let $\overrightarrow{v}$ be a vector normal to edge $e$, passing through point $y'$ to some point in face $f'$.
The initiated rays get critically reflect back into face $f$ from $e$ while making an angle $-\theta_c(f, f')$ with $\overrightarrow{v}$.  
Note that the rays in the discrete wavefront originating from $\kappa$ are parallel to each other; further, to achieve an $(1+\epsilon)$-approximation, the distance between any two such successive rays along $e$ is upper bounded by $\epsilon'$ (here, $\epsilon'$ is defined as in Lemma~\ref{lem:refrnoncritical}).

Let $y_1$ be the point located on $e$ at Euclidean distance $\Vert y'v'' \Vert - \epsilon'$  from $y'$ such that $y_1$ is located between $y'$ and $v''$.
Let $r$ and $r'$ be two critically reflected (parallel) rays (making an angle $-\theta_c(f, f')$ with $\overrightarrow{v}$), that originate from points $y'$ and $y_1$ respectively.

If both the rays are incident to the same edge, say $e''$ of $f$, we set $r, r'$ as a sibling pair and all the rays between $r$ and $r'$ that potentially cross the same edge sequence in the future together with $r, r'$ is a ray bundle. 
(See Fig. \ref{fig:initcritseg}.)
Let $x$ (resp. $x''$) be the point at which the ray $r$ (resp. $r'$) is incident onto edge $e''$. 
This sibling pair is traced further across the polygonal domain.

If the rays in a sibling pair $r, r'$ are incident to distinct edges of $f$, we need to find two rays $r_1', r_2'$ in $\calR(\kappa)$, to respectively pair up with $r$ and $r'$, for forming two new sibling pairs which define two corresponding ray bundles.
(See Fig. \ref{fig:splitsib-critseg}.)
Further, we also need to store an event point in the event heap that corresponds to the current shortest distance to $v$ if it is via the rays $r_1'$ and $r_2'$. 
At that stage, a discrete wavefront from $v$ is initiated. 
The algorithm to split a sibling pair that originates from a critical segment is described in Subsection \ref{subsubsect:splitcritseg}.

\subsubsection{Initiating rays from a critical source}
\label{subsect:initrayscritsrc}

There are three cases to consider based on the origin of the sibling pair, $r_1$ and $r_2$, of a bundle $B$ that represents  the discrete wavefront from the critical source $y'$ on an edge $e$:

\begin{enumerate}[(i)]
\item $r_1, r_2$ have originated from some critical segment
\item $r_1, r_2$ have originated from some vertex
\item $r_1, r_2$ have originated from (possibly distinct) nodes of a tree of rays, say $\calT_{R}(w)$
\end{enumerate}

We first rule out case (i) as it  cannot occur, since from  Proposition~\ref{prop:betwcrit}, 
a critical point of exit and a critical point of entry cannot occur in succession along a geodesic path.
Since case (iii) is a generalization of case (ii), herewith we explain event handling for case (iii).

Let $u \in \calT_{R}(w)$ be the origin of $r_1$.
Note that $u$ may be $w$ itself or a distinct critical source in $\calT_{R}(w)$.
Let $r,r'$ be successive rays in $\calR(u)$ such that $r$ is refracted and $r'$ is critically reflected.
Let $y'$ be the point at which $r'$ is incident to $v'v''$.
A node corresponding to critical source $y'$ is inserted as a child of $u$ in $\calT_{R}(w)$; 
further, a set $\calR(y')$ of rays are initiated from $y'$. 
Let $\theta'$ be the angle at which $r$ refracts (measured with respect to the vector $\overrightarrow{v}$ normal to edge $v'v''$).
Let $v'v''v'''$ be the face $f'$ onto which ray $r$ is refracted into.
Also, let $\overrightarrow{v_1}$ and $\overrightarrow{v_2}$ be two vectors in face $f''$ with origin $y'$ such that they respectively make $\frac{\pi}{2}$ and $\theta'$ angles with respect to $\overrightarrow{v}$. 
(See Fig. \ref{fig:initcritsrc}.)
The rays in $\calR(y')$ are uniformly distributed in the cone $C(y', \overrightarrow{y'v''}, \overrightarrow{v_2})$.
A ray bundle and a sibling pair corresponding to that ray bundle are determined:
Let $r_3$ be the ray in $\calR(y')$ that subtends the minimum required angle, $\epsilon'$ with edge $e$.
And let $r_2$ be the ray that originates from $y'$ and is parallel to ray $r$ in $f''$. 
Then $r_2, r_3$ are extremal pairs of rays in $\calT_{R}(w)$.
The bundle $B$ is now modified to  comprise all the rays that lie between $r_2$ and $r_3$, with $r_2$ and $r_3$ together forming a sibling pair of $B$.
This sibling pair is traced over the face $f''$.

When traced, if both $r_2$ and $r_3$ are incident to same edge of face $f''$, then they together will continue to be a sibling pair.
Otherwise, to form two sibling pairs, new rays to pair up with $r_2$ and $r_3$ are found from the sets of rays initiated in $\calT_R(w)$; the procedure is detailed in Subsection \ref{subsubsect:splittreeofrays}.
These sibling pairs' are traced and the corresponding event points are pushed to the event heap.

\subsection{Extending rays  across a region}
\label{subsect:extendrays}

 Given the bundles of rays that strike an edge $e=(a,b)$, the rays need to be extended across the region they enter.
 Let $T=(a,b,c)$ be the triangular region. The rays that refract after striking $e$, as well as the critical rays that 
 originate from $e$, need to be extended to determine their strike points on the other two edges $(a,c)$ and $(b,c)$.
 
 Extending a bundle may involve splitting a sibling pair. 
 While the split operation details will be defined subsequently,
 we discuss how to manage the extension of the set of bundles. 
 Note that only one of the bundles that strike $(a,b)$, 
 will lead to a weighted shortest path from the origin $s$ to the vertex $c$ from amongst the bundles that strike edge $(a,b)$. 
Furthermore, from a pair of bundles that cross each other, only one will be retained, due to the non-crossing property of shortest paths; and, from a  pair of bundles that split at vertex $c$ only one bundle need be split, i.e., the one that determines the shortest distance to $c$.

 Determining a weighted shortest path to $c$ via rays contained in a bundle $B$ will be detailed later. 
 Given a weighted shortest path to $c$ via each bundle in the current set of bundles ${\cal B}$ that strike $e$, one can determine a weighted shortest path to $c$ with respect to ${\cal B}$. The bundle $B' \in {\cal B}$ that determines
 a weighted shortest path is maintained. Note that ${\cal B}$ changes as bundles strike edge $e$.
 The bundle $B'$ partitions the bundles in ${\cal B}$ into two sets of bundles, one that strike
the edge  $(a,c)$ and the other set that strikes  $(b,c)$. 
 These bundles are traced after processing their strike on their corresponding edge.
 To ensure that the relevant bundles in ${\cal B}$ are traced across an edge, say $e= (a,c)$, we determine a shortest path to the edge $e$ via rays in bundle $B \in {\cal B}$.
 When this event occurs, the bundle is traced across edge $e$ if it does not violate non-crossing property. 
 A weighted shortest path to the edge is determined by a binary search over the space of rays in the bundle. 
This binary search procedure is similar to the process of determining a weighted shortest path to a specific vertex by a binary search.
%\todo{may need more details - LATER}
 
Note that the above  methodology is true for all bundles, independent of the source being a vertex or a critical segment.
We next determine how to efficiently determine the splitting of sibling pairs that define a bundle.

\subsection{Splitting a sibling pair}
\label{subsect:splitsibpair}

When a sibling pair needs to be split i.e.,  a ray bundle needs to be partitioned into two, then, based on the origin of the two rays in the sibling pair being considered, the following procedures are invoked.

\subsubsection{Pair that originates from a critical segment}
\label{subsubsect:splitcritseg}

\begin{figure}[h]
\begin{minipage}[t]{\linewidth}
\begin{center}
\includegraphics[totalheight=1.1in]{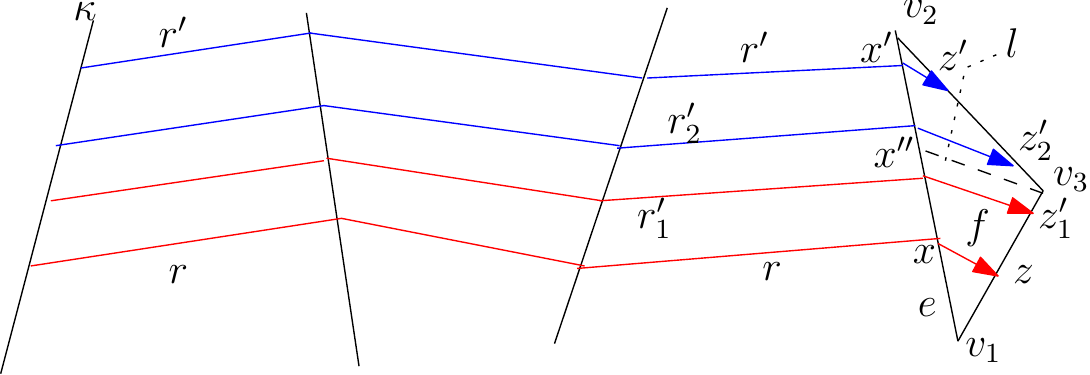}
\caption{\footnotesize Illustrating a ray bundle split when that ray bundle is originated from a critical segment $\kappa$; two new sibling pairs are also shown}
\label{fig:split-origcrigseg}
\end{center}
\end{minipage}
\vspace*{-0.2in}
\end{figure}

Consider a sibling pair $r, r'$ that originates from a critical segment $\kappa$. 
(See Fig. \ref{fig:splitsib-critseg}.)
Let $\calE$ be the edge sequence of $r$ and $r'$ that the two rays have in common till they are incident to an edge $e$ of face $f$ and let $B$ be the corresponding ray bundle.
Let $r$ be incident to $x \in e'$ and let $r'$ be incident to $x'' \in e''$.
Given that $r$ and $r'$ are parallel between any two successive edges in $\calE$, following are the possibilities:
$r$ and $r'$ refract from $x$ and $x''$ respectively at the same angle;
$r$ and $r'$ critically reflect from $x$ and $x''$.
Since the latter kind of rays are not geodesic paths (Proposition \ref{prop:betwcrit}), we focus on the former.
If the rays refracted from $x$ and $x''$ are both incident to the same edge of $f$, then there is no need to split the sibling pair $r, r'$.
Otherwise, as explained below, we find new rays $r_1'$ and $r_2'$ that originate from $\kappa$ to respectively pair with $r$ and $r'$.

Let $v_1, v_2, v_3$ be the vertices of $f$ where the endpoints of $e$ are $v_1, v_2$. 
(See Fig. \ref{fig:split-origcrigseg}.)
Let $l$ be a line segment on face $f$ such that it is parallel to line segment $f \cap r$ ($r$ refracted at $p$) and passes through vertex $v_3$ of face $f$. 
Let $x'' \in v_1v_2$ be its other endpoint.
We interpolate over rays in $\calR(\kappa)$ to find a ray $r_1' \in \calR(\kappa)$ (resp. $r_2' \in \calR(\kappa$)) such that the section of $r_1'$ (resp. $r_2'$) on $f$ lies between $l$ and the section of $r$ (resp. $r'$) on $f$.
Then the ray bundle $B$ splits into: a bundle with sibling pair $r, r_1'$ and another with $r', r_2'$ as its sibling pair.
Since between any successive edges in $\calE$, rays $r$ and $r'$ are parallel, an interpolation in possible.
Based on the ratio of $\Vert xx'' \Vert$ to $\Vert xx' \Vert$, we interpolate to find a point $q$ on $\kappa$ so that a critically reflected ray $r'''$ from $q$ reaches $v_3$.
The ray $r_1'$ (resp. $r_2'$) is the one whose origin is closest to $q$ among all the rays between $r'''$ and $r$ (resp. $r'$).
These sibling pairs' are pushed to the event heap with the key values being the respective weighted Euclidean distances' from $s$ to the points at which these pairs strike the edges $v_1v_3$ and $v_2v_3$. 

Let the ray $r_1'$ be incident to $v_1v_3$ at $z_1'$ and let $r_2'$ be  incident to $v_2v_3$ at $z_2'$.
An event point to initiate a discrete wavefront from $v_3$ is pushed  to the heap with the key value $\min(d_{z_1'} + w_{v_1v_3} \Vert z_1'v_3 \Vert, d_{z_2'} + w_{v_2v_3} \Vert z_2'v_3 \Vert)$ where $d_{z_1'}$ is the weighted Euclidean distance from $s$ to $z_1'$ and $d_{z_2'}$ is the weighted Euclidean distance from $s$ to $z_2'$.

\subsubsection{Pair that originates from a tree of rays}
\label{subsubsect:splittreeofrays}

Let $e_i, e_j, e_k$ be the edges bounding face $f$.
Also, let $r', r''$ be a sibling pair of $\calT_R(w)$ such that both of them strike edge $e_i$ of $f$.
With further expansion of the wavefront, let $r'$ strike edge $e_j$ at $q'$ and let $r''$ strike edge $e_k$ at $q''$.
This requires us to split the ray bundle, $B$, corresponding to sibling pair $r', r''$.
Let $P'$ (resp. $P''$) be the critical ancestor path of $r'$ (resp. $r''$).
Further, let $REG$ be the open region bounded by rays along critical ancestor paths $P', P''$, the rays $ r', r''$ and the line segment $q'q''$.
Let $R_1$ (resp. $R_2$) be the set of rays such that a ray $r \in R_1$ (resp. $r \in R_2$) if and only if $r$ originates from a critical point of entry or a critical source located on the critical ancestor path $P'$ (resp. $P''$) and the ray $r$ lies in $REG$.
With binary search over the rays in $R_1$, we find a ray $r_1 \in R_1$ and with binary search over $R_2$ we find a ray $r_2$ such that $r_1$ intersects $e_j$, $r_2$ intersects $e_k$, and $r_1$ and $r_2$ are either successive rays originating from the same origin or adjacent origins on a critical ancestor path. The binary search
is performed over the nodes on the critical path, and at each node $v$ the two extreme rays in the set of rays $\calR(v)$ that originate  at that node are used to decide whether the two rays $r_1$ and $r_2$ lie within $\calR(v)$ or not. 
This leads to splitting ray bundle  $B$ into two $B_1$ and $B_2$ 
with $r_1, r'$ and $r'', r_2$ sibling pairs, respectively.
Events corresponding to the ray bundles created are pushed to the event queue reflecting the split.
Further, we also push the event corresponding to initiating a discrete wavefront from $v$.
Note that splitting a sibling pair that originates from a vertex is just a special case of the procedure listed above. 

\section{Improving the time complexity: interpolating versus tracing rays}
\label{sect:interpol}

To improve the time complexity of our algorithm while obtaining a weighted shortest path with $(1+\epsilon)$ multiplicative error, instead of tracing a ray $r$ across an edge sequence to find its point (and angle) of incidence onto an edge $e$, we show that it suffices to interpolate the position and the angle of refraction of $r$ from the last edge of $\calE$.
For this, we need an assumption (albeit mild) that the angle between any two successive rays in any ray bundle is upper bounded by $\pi - \frac{1}{n^n}$.

\begin{wrapfigure}{r}{0.5\textwidth}
\centering
\begin{minipage}[b]{.4\textwidth}
\centering{\includegraphics[totalheight=1.5in]{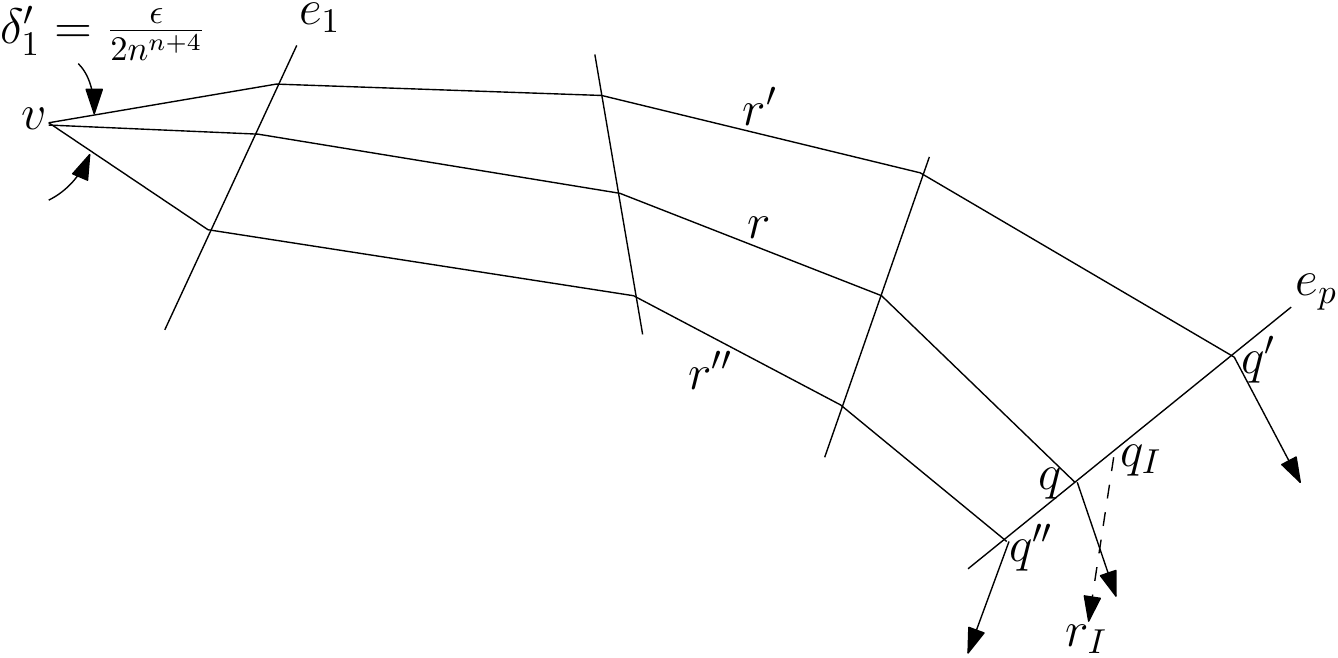}}
\caption{\footnotesize Illustrating the error in interpolating versus tracing a ray across an edge sequence $\calE$}
\label{fig:interptrace}
\end{minipage}
\end{wrapfigure}

From Lemma~\ref{lem:depsilon} and Lemma~\ref{lem:refrcritical}, we note that in order to find an $(1+\epsilon)$-approximate weighted shortest path, the cardinality of the set $S$ of rays that originate from a vertex, or a critical source, is chosen to be $O(\frac{\mu}{(\epsilon')^{n^2+1}})$.
Let $\calE$ be an edge sequence of a sibling pair $r', r''$, and let $e$ be the last edge of $\calE$.
In doing binary search over $S$ to find a ray $r \in S$, we need to trace $O(\lg(\frac{\mu}{\epsilon'}) + n^2\lg({\frac{1}{\epsilon'}}))$ rays in $S$ to $e$.
Taking into account the cardinality of $\calE$ (which is $O(n^2)$ from the Proposition~\ref{prop:edgeseqlen}), the time it takes for binary search to find a ray from the source is $O(n^2 \lg(\frac{\mu}{\epsilon'}) + n^4\lg({\frac{1}{\epsilon'}}))$.
As detailed below, we reduce the time involved in this by interpolation. 

%Let $r'$ and $r''$ be two adjacent traced rays in a bundle $B$ of rays in ${\calT}_R(v)$, with the smallest angle between them at their source $v$ is $\delta'_1$.
Let $r'$ and $r''$ be two sibling rays in a bundle $B$ in ${\calT}_R(v)$, with $\delta'_1 (= \frac{\epsilon}{2n^{n+4}})$ being the angle between them.
(See Fig. \ref{fig:interptrace}.)
Consider a ray $r \in B$ with origin $v$ and that lies between $r'$ and $r''$.
Also, let $e_p$ be an edge that is intersected by all three rays $r, r',$ and $r''$, respectively at points $q, q'$ and $q''$.
We denote $\frac{||q'q||}{||q'q''||}$ with $\gamma$.
(This ratio is known if the target point $q$ on $e$ is known even though $r$ itself may not be known.)
The interpolation of ray $r$ is denoted with $r_I$ and it is characterized as follows:

\begin{wrapfigure}{r}{0.5\textwidth}
\centering
\begin{minipage}[b]{.4\textwidth}
\centering{\includegraphics[totalheight=1in]{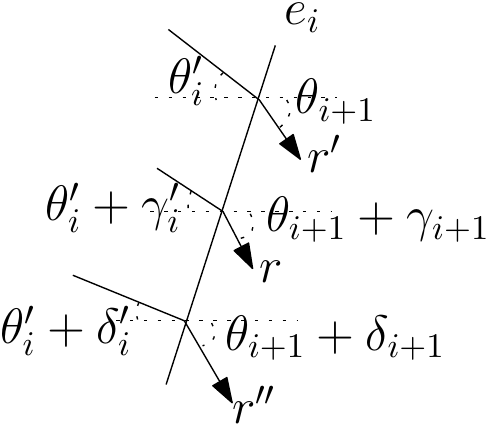}}
\caption{\footnotesize Illustrating the angle of incidence and refraction of rays $r', r''$, and $r'$ at an edge $e_i \in \calE$}
\label{fig:interpolangles}
\end{minipage}
\end{wrapfigure}

\begin{itemize}
\vspace{-0.1in}
\item[(i)]
Let $f_p$ be the face bounded by $e_{p-1}$ and $e_p$.
The angle between the vectors induced by $f_p \cap r_I$ and $f_p \cap r'$ is upper bounded by $\delta_1'\gamma$.

\item[(ii)]
The Euclidean distance of interpolated point of incidence $q_I$ of ray $r_I$ from $q'$ is $q' + \gamma \Vert q'q'' \Vert$.
\vspace{0.05in}
\end{itemize}

\begin{lemma}
\label{lem:interpolate}
Let $\calE$ be the edge sequence associated with a ray bundle $B$ in ${\calT}_R(v)$, having siblings $r'$ and  $r''$; and, let $e$ be the last edge in $\calE$.
Let $r \in B$ be a ray that lies in between $r'$ and $r''$ and let $q$ be the point of incidence of $r$ on edge $e$ when traced across $\calE$.
Also, let $q_I$ be the point due to interpolation. 
Given that the angle $\gamma_1$ between $r'$ and $r''$ is upper bounded by $\frac{\epsilon}{2n^{n+4}}$, the weighted distance $d$ between $q_I$ and $v$ is approximated by $r_I$ within a multiplicative error of $(1+\frac{\epsilon}{n^2})$, and the angle of incidence of $r_I$ is $(1+\epsilon)$-approximated with respect to $r$.
\end{lemma}

\begin{proof}
Let $e_1, e_2, \ldots, e_p$ be the edge sequence in $\cal{E}$.
As in Lemma~\ref{lem:depsilon}, angles of incidence and angles of refractions of rays $r'$ and $r''$ are as mentioned here.
Let $\theta_1', \theta_2', \ldots, \theta_{p}'$ be the angles at which the ray $r'$ is incident on edges $e_1, e_2, \ldots, e_{p}$ respectively.
(See Fig. \ref{fig:interpolangles}.)
Let $\theta_2, \theta_3, \ldots, \theta_{p+1}$ be the angles at which the ray $r'$ refracts at edges $e_1, e_2, \ldots, e_{p}$ respectively.
Similarly, let $\theta_1'+\delta_1', \theta_2'+\delta_2', \ldots, \theta_{p}'+\delta_{p}'$ be the angles at which the ray $r''$ is incident on edges $e_1, e_2, \ldots, e_{p}$ respectively.
And, let $\theta_2+\delta_2, \theta_3+\delta_3, \ldots, \theta_{p+1}+\delta_{p+1}$ be the angles at which the ray $r''$ refracts at edges $e_1, e_2, \ldots, e_{p}$ respectively.
For every $i$, if a critical angle exists for edge $e_i$, assume that both $\theta_i'$ and $\theta_i' + \delta_i'$ are less than that critical angle.
Using Lemma~\ref{lem:depsilon}, we derive a bound on $\delta_p$ as follows: 

\begin{eqnarray*}
\delta_p &=& \frac{w_1}{w_p} \delta_1' \prod_i
\frac{\cos(\theta_i'+\delta_i')}{\cos(\theta_{i+1}+\delta_{i+1})} \hspace{0.3in} \text{(from (\ref{eq:deltan}))}\\
&\le& \frac{w_1}{w_p} \delta_1' \prod_i
\frac{\cos(\theta_i')-\delta_i'(\sin(\theta_i'))}{\cos\theta_{i+1}'-\delta_{i+1}\sin\theta_{i+1}}\\
&\le& \frac{w_1}{w_p} \delta_1' \prod_i
\frac{\cos{\theta_i'}}{\cos{\theta_{i+1}}-\delta_{i+1}} \\
&= &\frac{w_1}{w_p} \delta_1' (\prod_i \frac{\cos\theta_i'}{\cos\theta_{i+1}})
(\prod_j(1-\frac{\delta_j}{\cos{\theta_j}})^{-1})\\
&\le &\frac{w_1}{w_p} \delta_1' (\prod_i
\frac{\cos\theta_i'}{\cos\theta_{i+1}}) (\prod_j
(1+\frac{\delta_j}{\cos\theta_j}+ o(\frac{\delta_j}{\cos\theta_j}))\\
&\le& \frac{w_1}{w_p} \delta_1' (\prod_i
\frac{\cos\theta_i'}{\cos\theta_{i+1}})
(1+\sum_j\frac{\delta_j}{\cos\theta_j} +
o(\sum_j \frac{\delta_j}{\cos\theta_j}))\\
\end{eqnarray*}

Let $\theta_1'+\gamma_1', \theta_2'+\gamma_2', \ldots, \theta_{p}'+\gamma_{p}'$ be the angles at which the ray $r$ is incident on edges $e_1, e_2, \ldots, e_{p}$ respectively.
Also, let $\theta_2+\gamma_2, \theta_3+\gamma_3, \ldots, \theta_{p_1}+\gamma_{p+1}$ be the angles at which the ray $r$ refracts from edges $e_1, e_2, \ldots, e_p$ respectively.
Then analogous to the above, we can lower bound the angle at which ray $r$ strikes edge $e_j$ as follows: 
$\gamma_p \ge \frac{w_1}{w_p} \gamma_1' (\Pi_i\frac{\cos\theta_i'}{\cos\theta_{i+1}})
(1-\sum_j\frac{\gamma_j}{\cos\theta'_j} + o(\sum_j\frac{\gamma_j}{\cos\theta'_j}))$.
%$\gamma_p \le \frac{w_1}{w_p} \gamma_1' (\prod_i\frac{\cos\theta_i'}{\cos\theta_{i+1}}) (1+\sum_j\frac{\gamma_j}{\cos\theta'_j} + o(\sum_j\frac{\gamma_j}{\cos\theta'_j}))$.
For every $2 \le i \le p$, let $f_i$ be the face bounded by $e_{i-1}$ and $e_i$.
Also, for every $2 \le i \le p$, let $\delta_j^I$ be the angle between the vectors induced by $f_j \cap r_I$ and $f_j \cap r'$.
The angle $\delta_p^I$ is upper bounded by 
$\frac{w_1}{w_p} \delta_1' \gamma(\prod_i
\frac{\cos\theta_i'}{\cos\theta_{i+1}})
(1+\sum_j\frac{\delta^I_j}{\cos\theta_j} +
o(\sum_j \frac{\delta^I_j}{\cos\theta_j})$.
Hence, the ratio of the  error in the angle $\delta^I_p$ versus the angle $\gamma_p$, denoted with $R_p$ equals to  
\begin{eqnarray*}
& & \frac{\delta^I_p - \gamma_p}{\gamma_p} \\ 
&\le& \frac{(\sum_i \frac{\delta_i\gamma}{\cos\theta_i}- \frac{\gamma_i}{\cos\theta'_i}) +
o(\sum_j\frac{\delta_j}{\cos\theta_j}  ))}{1+\sum_j
\frac{\gamma_j}{\cos\theta'_j}}  \\
&\le& (\sum_i \frac{2\delta_i}{\cos\theta_i}) + o(\sum_i
\frac{\delta_i}{\cos\theta_i})  \\
&\le& n^2 \max_i \frac{2\delta_i}{\cos{\theta_i}} + o(\sum_i
\frac{\delta_i}{\cos\theta_i}) \  \   \mbox{ since}  \ p \leq n^2\\
&\le& \frac{\epsilon}{n^2} \quad  \   \mbox{when}   \  \max_i \frac{2\delta_i}{\cos{\theta_i}} \le
\frac{\epsilon}{n^4}
\end{eqnarray*}
Therefore, to keep the error ratio less than $\epsilon$, we choose $\delta_1' \le \frac{\epsilon}{2n^{n+2}}$ since $\min_i \cos{\theta_i} \geq \frac{1}{n^n}$. 
(In the last inequality, we used the assumption that $\theta'_{i} \leq \frac{\pi}{2} - \frac{1}{n^n}$).
In fact, we choose $\delta_1'$ less than or equal to $\frac{\epsilon}{2n^{n+4}}$, so that to upper bound $R_p$ with $\epsilon$ as well as to upper bound $\Vert q_I-q \Vert$.

%Let $q_1, q_2, \ldots, q_p=q$ be the points of incidence of ray $r$ with sequence of edges $e_1, e_2, \ldots, e_p$ in $\calE$.
Also, let $q'_1, q'_2, \ldots, q'_p=q_I$ respectively be the points of incidence of interpolated ray $r_I$ on edges $e_1, e_2, \ldots, e_p$ in $\calE$.
%Suppose the path traversed by the ray $r$ is a piecewise linear sequence  with  parts of lengths, say $d_1, d_2 \ldots d_{p}$. 
%For the $i^{th}$ part of the piece-wise sequence, the error ratio is $ R_i = \frac{E_i}{\gamma_i}$, where $E_i$ is the difference  in the angle at $e_i$ between the interpolated ray and $r$.
Considering that the ray $r$ is traced across the edge sequence of cardinality $O(n^2)$, the $\Vert q_Iq \Vert$ is bounded recursively.
Let $dist_r$ be the distance along ray $r$ from $v$ to $q$ and let $dist_{r_I}$ be equal to $\sum_{i=2}^{p} \Vert q'_{i-1}q'_i \Vert$.
Since the error in angle accumulates, $\frac{dist_{r_I} - dist_r}{dist_r}$ is upper bounded by $\frac{\epsilon}{n^2}$.
\end{proof}

In the interpolation method, the binary search is performed on the ordered set of rays, $\calR (v)$, at vertex $v$ until the angle between the rays guiding the binary search is less than or equal to $\frac{\epsilon}{n^{n+4}}$.
This results in $O(n (\lg \frac{n}{\epsilon}))$ search steps, where each step would require tracing a ray from the source.
%The interpolation method defined above relies on the knowledge of three points $q, q'$ and $q''$. 
Moreover, this method can also be used to approximate the shortest distance to an edge via rays in a ray bundle.
%Given a range of points on edge $e$ intersected by the bundle $B$, with $q'$ and $q''$ the two extreme points, a binary search can be used over the segment $q'q''$ and the interpolation method used to approximate the shortest distance to the point guessed during binary search.  
%This would determine the approximate shortest distance to edge $e$. 
We summarize the discussion in the following Lemma.
\begin{theorem}
\label{thm:interpol}
The interpolation method determines the shortest distance to a given point $p$ from a source $v$ to within a multiplicative approximation factor of $(1+\epsilon)$ in $O(n\lg(\frac{n}{\epsilon}))$ time using rays from a ray bundle. 
Furthermore, an $(1+\epsilon)$-approximation to the shortest distance from $v$ to an edge $e$  can be determined in $O(n(\lg{\frac{n}{\epsilon}})(\lg{\frac{\mu}{\sqrt{\epsilon}}}))$ time using rays from a ray bundle.
\end{theorem}

\section{Analysis}
\label{sect:analysis}

The analysis is based on showing the following facts:
(i) bundles of rays are correctly maintained, and
(ii) the shortest distances to vertices, and edges similarly, are correctly computed using the bundles. 

Bundles are initiated from vertices and propagated across the faces of the domain as they strike the 
edges of the domain. The propagation is evidently correct, the modifications to the bundles
being (a) bundle splits and (b) elimination of bundles due to crossing of bundles.
As described in Subsection \ref{subsect:extendrays}, elimination of bundles is
determined at vertices when more than one of the bundles split at that vertex.
Note that bundles are propagated until they split and further propagation of the split parts of the bundles occurs when the shortest distance event, of weighted distance $d$, corresponding to the bundle, say $B$,  striking an edge is determined by the heap. 
When two bundles cross each other completely, one of them is eliminated from further consideration.
Thus for the rest of the proof, we will assume w.l.o.g., that bundles are correctly maintained.

To determine approximate weighted shortest paths to vertices correctly, we consider two categories of paths, Type-I and Type-II. 
Recall that Type-I paths do not have critically reflected segments in between, while Type-II paths does have. 
Analogous arguments prove the correctness involved in computing approximate weighted shortest paths to edges.

\begin{lemma}
\label{lem:type1find}
The algorithm correctly determines sibling pair in $\calT_R (u)$ for every source $u$, and computes an $(1+\epsilon)$-approximation to a Type-I weighted shortest path from $u$ to $v$, for any vertex $v$ in $\mathcal{P}$.
\end{lemma}
\begin{proof}
Suppose that there exists a Type-I weighted shortest path between $u$ and $v$, traversing an edge sequence ${\calE}$.
Lemma~\ref{lem:refrnoncritical} shows that an $(1+\epsilon)$ approximation can be found using rays within
a bundle. Bundles are maintained using sibling pair.
We show that the algorithm maintains sibling pair for edges in ${\cal E}$ that determine bundles, the sibling pairs being traced rays such that these rays can be used to find an $(1+\epsilon)$-approximate weighted shortest path from $u$ to $v$.

Initially when rays are generated from a vertex, $v$, siblings are computed correctly and each ray is refracted correctly.
Consider the procedure that computes a sibling pair $r_1, r_2$ closest to a vertex $v$.
The rays $r_1$ and $r_2$ are thus two successive rays and the distance to $v$ via rays $r_1$ and $r_2$ is updated.
Suppose both the rays $r_1$ and $r_2$ have the same origin.
Using binary search with interpolation, we find the siblings correctly.
Otherwise, the origin of $r_1$ is different from the origin of $r_2$.
In this case, we find a sibling pair via the binary search with interpolation on the rays in the critical ancestor paths of two siblings, say $r'$ and $r''$ such that the rays $r_1$ and $r_2$  are part of the set of rays lying in between $r'$ and $r''$.
The correctness of the interpolation method is shown in Lemma~\ref{lem:interpolate}.
\end{proof}

\begin{lemma}
\label{lem:refrcritical-find}
Let $P$ be a Type-II weighted shortest path from a vertex $v$ to another vertex $w$ on $\calP$ with a critical segment ${\kappa}$ in-between.
Then the algorithm determines a pair of traced rays  in $\calR(v)$ that can approximate a weighted shortest path, $P$. 
\end{lemma}
\begin{proof}
An optimal path $P$ can be partitioned into sub-paths, each such sub-path going from one vertex $v$ to another vertex $w$. 
Suppose a sub-path uses a critical segment $\kappa$ and is partitioned as follows: a path $P_1$ from $v$ to $\kappa$ and a path $P_2$ from $\kappa$ to $w$.
Let $e$ be the edge on a face $f$ such that the critical segment $\kappa$ lies  on $e$ and reflects rays back onto face $f$. 
Also, let $y$ be the critical point of entry into $\kappa$.
The correct determination of the critical point of incidence from a bundle of rays follows from
Lemma~\ref{lem:refrnoncritical}.
If $w$ is the endpoint of $\kappa$ then we are done since the distance to the endpoint from the critical source is included in consideration.
Otherwise rays are generated from  $\kappa$ that are parallel and separated by a small weighted Euclidean distance  less than $\epsilon'$.
Let $r_1, r_2$ be sibling rays that originate from $\kappa$ such that $w$ lies in between $r_1$ and $r_2$.
Lemma~\ref{lem:refrcritical} shows that the shortest distance from $\kappa$ to $w$ can
be found by tracing the rays $r_1$ and $r_2$ and interpolating between them to find the  point
on $\kappa$ closest to $w$.
\end{proof}

\begin{theorem}
\label{thm:timecompl} 
The algorithm computes an $(1+\epsilon)$-approximate weighted shortest path from $s$ to $t$ in 
%$O(n^4(n\lg{n} + \lg(\frac{\mu}{\epsilon}(1+\frac{1}{\sin{\theta_{min}}}))))$ time.
$O(n^5(\lg{\frac{n}{\epsilon}})(\lg{\frac{\mu}{\sqrt{\epsilon}}}))$ time. 
\end{theorem}

\begin{proof}
The correctness follows from Lemmas~\ref{lem:type1find} and \ref{lem:refrcritical-find}.

We first bound the number of ray bundles.
The number of bundles that are initiated from vertex sources are $O(n)$. 
By Proposition 3, the number of bundles that are initiated from critical segments is bounded by $O(n^2)$.
The processing of such bundles will be handled separately.

The number of bundles increase when split. 
The increase is charged to the vertex $v$ that causes the split. 
Only one bundle, the bundle that determines an approximate weighted shortest distance to $v$ is split and thus the total number of bundles split are $O(n)$.
Bundles are propagated further into adjacent faces after striking an edge.

We thus need to analyze the time involved in finding an approximate weighted shortest paths and splitting ray bundles.
Let $S$ be the set of rays from a vertex source $v_1$.
To search for  a  ray in $S$ that is closest to the vertex or provides the shortest distance to the last edge $e$ of an edge sequence $\calE$, the algorithm utilizes the combination of binary search and interpolation (as described in Section~\ref{sect:interpol}).
Using Theorem~\ref{thm:interpol}, the complexity of determining the shortest distance  when  rays from a specific vertex are considered
is equal to $O(n^3 \lg (\frac{n}{\epsilon}) \lg (\frac{\mu}{\sqrt{\epsilon}}))$.
Note that $O(n^2)$ edges will be encountered as rays are traced.
%The same asymptotic worst-case time is spent when the ray bundle splits are considered at critical points of entry.

Next, let us consider the  case in which the sibling pair is not from a vertex but arises from a  critical source
contained in a tree of rays, say $\calT_R(u)$.
Let $S$ be the set consisting of all the critical points of entry in any critical ancestor path $P$ in $\calT_R(u)$.
A binary search on the critical path determines vertex $u$ such that the ray to be determined is in $\calR(u)$.
Since $S$ is $O(n^2)$ in size and since it is ensured that no two ray bundles cross, binary search can be used to find a vertex $v$  in $\calT_R(u)$ such that an approximate weighted shortest path lies in $\calR(v)$ with $O(\lg (n))$ ray tracings.
A search on the rays from that specific vertex in $\calR(u)$ follows as described in the earlier paragraph. 

A similar search is required to determine the ray that identifies the approximate weighted shortest distance to an edge, which is required when a bundle splits.
The procedure and time complexity of this is similar to the determination of a weighted shortest path to a vertex. 

Since the total number of vertices and critical sources is $O(n^2)$, the total work in splitting and initiating ray bundles from these sources is $O(n^5 \lg (\frac{n}{\epsilon}) \lg (\frac{\mu}{\sqrt{\epsilon}}))$, as claimed in Theorem~\ref{thm:interpol}.

Determining the ray originating from a critical segment and that strikes  a vertex can be determined via interpolation (or a binary search) on the space of the parallel rays that are part of the bundle that originates from a  critical segment.
The work involved in tracing a ray that originates at a critical segment and traverses an edge sequence takes $O(n^2)$ time (Proposition~\ref{prop:edgeseqlen}).
The specific pair of successive rays of interest can be found by interpolation, taking $O(1)$ time.
As there are $O(n)$ vertices and $O(n^2)$ critical segments (Proposition~\ref{prop:numcritsrc}), the time complexity is $O(n^5\lg(\frac{\mu}{\sqrt{\epsilon}}))$, including the binary search involved.
%
%The total time complexity bound follows.
\end{proof}

\subsection*{Single-source approximate weighted shortest path queries}
\label{subsect:sssp}

Here we devise an algorithm to preprocess $\cal{P}$ to construct a {\it shortest path map} so that for any given query point $q$ in $\cal{P}$ an approximate weighted shortest path from $s$ can be computed efficiently.
As part of preprocessing, using the non-crossing property of shortest paths (Proposition~\ref{prop:noncrossing}), we compute for every edge $e$ in the subdivision a minimum cardinality set ${\cal B}(e)$ of bundles such that for every point $p \in e$ an approximate weighted shortest path can be computed from a bundle $B \in {\cal B}(e)$.
Further, for every edge $e$, the set ${\cal B}(e)$ of bundles are saved with $e$.
For a query point $q \in {\cal P}$, exploiting Lemma~\ref{lem:errorX} and Corollary~\ref{cor:errorX}, we locate $q$ in the triangulation and determine an approximate weighted shortest path from $q$ to $s$ via one of the bundles associated to edges of the triangle containing $q$.

As mentioned, let $N$ be the maximum coordinate value used in describing $\calP$.
Since the number of bundles is $O(n^2)$ and since it takes $O(n^2  (\lg{\frac{n}{\epsilon}})(\lg{\frac{\mu}{\sqrt{\epsilon}}}) (\lg{N}))$ time to determine the bundles that need to be associated to an edge and there $O(n)$ edges, it requires $O(n^5 (\lg{\frac{n}{\epsilon}}) (\lg{\frac{\mu}{\sqrt{\epsilon}}})(\lg{N}))$ time for preprocessing.
The query phase takes $O(n^4 (\lg{\frac{n}{\epsilon}})( \lg{\frac{\mu}{\sqrt{\epsilon}}})(\lg{N}))$ time: there could be $O(n^2)$ bundles associated to edges of the triangle containing $q$ and finding an approximate weighted shortest path via any one such bundle takes $O(n^2(\lg{\frac{n}{\epsilon}})(\lg{\frac{\mu}{\sqrt{\epsilon}}})(\lg{N}))$ time.

\section{Conclusions}
\label{sect:conclu}

In this paper we have presented a polynomial-time algorithm for finding an approximate weighted shortest path between two given points $s$ and $t$.
The main ideas of this algorithm rely on progressing the discretized wavefront from $s$ to $t$. 
The time complexity of our algorithm is $O(n^4(n \lg{n} + \text{ polylog}(\mu, \epsilon, \theta_{min})))$.
Significantly, our algorithm is polynomial with respect to input parameters.
This result is about a cubic factor (in $n$) improvement over the Mitchell and Papadimitriou's '91 result \cite{journals/jacm/MitchellP91} in finding a weighted shortest path between two given points, which is the only known polynomial time algorithm for this problem to date.
In addition, we extend our algorithm to answer single-source weighted shortest path queries.
Further, with minor modifications, our algorithm appears extendable to determine geodesic shortest paths on the surface of a $2$-manifold whose faces are associated with positive weights. 
Since the number of events in the problem stand at $\Omega(n^4)$ (from \cite{journals/jacm/MitchellP91}), it would be interesting to explore further improvements in devising a more efficient polynomial time approximation scheme.

\bibliographystyle{plain}

\bibliography{weireg-sp}

\end{document}